\documentclass[
aps,
pra,
10pt,
twocolumn,
groupedaddress,
amsmath,
amssymb,
hypertext
]{revtex4-2}

\usepackage{graphicx}

\usepackage{algorithmic}
\usepackage{xcolor}
\usepackage{mathrsfs}
\usepackage{mathtools}
\usepackage{nicefrac}
\usepackage{bbm}
\usepackage{braket}
\usepackage{booktabs}
\usepackage{placeins}
\usepackage{tikz}
\usetikzlibrary{calc,patterns,angles,quotes}
\usepackage{quantikz}

\usepackage{xcolor}
\definecolor{mygreen}{RGB}{205, 222, 194}

\DeclareMathOperator*{\argmin}{arg\,min}

\newtheorem{theorem}{Theorem}

\newtheorem{proposition}[theorem]{Proposition}
\newtheorem{corollary}[theorem]{Corollary}

\newtheorem{remark}[theorem]{Remark}

\newtheorem{problem}[theorem]{Problem}

\providecommand{\qedsymbol}{$\square$}
\providecommand{\qed}{\hfill\qedsymbol}
\newenvironment{proof}[1][Proof]{\noindent\textbf{#1.}\hspace*{0.5em}}{\qed\par\bigskip}

\newenvironment{proof-sketch}{\noindent\textbf{Proof Sketch}\hspace*{0.5em}}{\qed\bigskip}
\newenvironment{proof-idea}{\noindent\textbf{Proof Idea}\hspace*{0.5em}}{\qed\bigskip}
\newenvironment{proof-of-lemma}[1][{}]{\noindent\textbf{Proof of Lemma {#1}.}\hspace*{0.5em}}{\qed}
\newenvironment{proof-of-theorem}[1][{}]{\noindent\textbf{Proof of Theorem {#1}.}\hspace*{0.5em}}{\qed}
\newenvironment{proof-of-corollary}[1][{}]{\noindent\textbf{Proof of Corollary {#1}.}\hspace*{0.5em}}{\qed}



\usepackage{silence}
\makeatletter
\def\@bibdataout@aps{
 \immediate\write\@bibdataout{
  @CONTROL{
   apsrev42Control
   \longbibliography@sw{
    ,author="48",editor="1",pages="0",title="0",year="1"
   }{
    ,author="48",editor="1",pages="0",title="",year="1"
   }
  }
 }
 \if@filesw
  \immediate\write\@auxout{\string\citation{apsrev42Control}}
 \fi
}
\makeatother

\usepackage[
  bookmarks=true,
  bookmarksnumbered=true,
  colorlinks=true,
  linkcolor=blue,
  citecolor=blue,
  urlcolor=blue,
]{hyperref}

\begin{document}

\def\revtex{}

\title{Pauli-Structured Preconditioning for Quantum Linear System Solvers}
\author{Hantao Nie}
\email{nht@pku.edu.cn}
\affiliation{School of Mathematical Sciences, Peking University, Beijing, China}

\author{Zhijian Lai}
\email{lai\_zhijian@pku.edu.cn}
\affiliation{Beijing International Center for Mathematical Research, Peking University, Beijing, China}

\author{Dong An}
\email{dongan@pku.edu.cn}
\affiliation{Beijing International Center for Mathematical Research, Peking University, Beijing, China}

\date{\today}

\begin{abstract}
  Preconditioning is a fundamental technique for accelerating classical linear system solvers, and understanding when its benefits persist in quantum linear system (QLS) solvers is important for assessing the practical resource requirements of quantum linear algebra. In QLS algorithms, however, the potential advantage of preconditioning may be offset by the normalization overhead incurred by composing separate block-encodings of the system matrix and the preconditioner, as observed in recent work. This limitation leaves open whether additional algebraic structure can make preconditioning effective in quantum access models.
  Motivated by this question, we show that Pauli-structured representations of both the system matrix and the preconditioner allow the preconditioned operator to be accessed through regrouped Pauli expansions.
  In this setting, algebraic regrouping of Pauli products can reduce the Pauli coefficient weight of the preconditioned operator, thereby altering the normalization parameters relevant to quantum algorithms.
  We derive explicit size and coefficient-weight bounds for the regrouped Pauli representations, and we trace their consequences for both direct block-encoding constructions and randomized Pauli linear system solvers.
  These results identify when Pauli-structured preconditioning can reduce the effective complexity parameters of QLS algorithms, rather than merely improving the classical condition number.
  Numerical experiments on a finite-dimensional synthetic benchmark show reductions in norm-aware direct block-encoding diagnostics and in the randomized QLS per-sample depth proxy.
  \end{abstract}

\maketitle

\section{Introduction}
Preconditioning is a fundamental technique in numerical linear algebra for improving the conditioning of linear systems. Given $A \in \mathbb{C}^{N \times N}$ and $b \in \mathbb{C}^N$, the original system $A x=b$ is reformulated as an equivalent preconditioned system, either $PAx=Pb$ or, after the change of variables $x=P^{\dagger} y$, $PAP^{\dagger}y=Pb$, where $P \in \mathbb{C}^{N \times N}$ is chosen so that the resulting system is better conditioned.
This transformation can substantially accelerate iterative solvers~\cite{saad2003iterative,benzi2002preconditioning,grote1997parallel}. The quantum analog is the quantum linear system (QLS) problem~\cite{harrow2009quantum,childs2017quantum,low2017optimal,gilyen2019quantum,lin2020optimal,costa2022optimal,morales2024quantum}, whose complexity depends not only on the condition number but also on the oracle model used to access $A$ and the right-hand side $b$. This raises a basic question: \textit{under which access models can preconditioning genuinely reduce the query complexity of a QLS solver?}

Considerable effort has been devoted to quantum preconditioning.
Contributions include structural preconditioners tailored to specific problem classes~\cite{clader2013preconditioned,shao2018quantum,wan2018asymptotic,golden2022quantum}, matrix function and fast inversion techniques~\cite{tong2021fast}, variational formulations~\cite{hosaka2023preconditioning}, optimization-oriented pipelines~\cite{wu2024preconditioned}, annealing-based sparse approximate inverses~\cite{suresh2023sparse}, and several partial differential equation (PDE)-motivated constructions~\cite{phillips2021quantum,jin2025schrodingerization,shayegan2025quasi}.  Moreover, a comprehensive survey of the area appears in~\cite{raisuddin2024review}. A recurring observation across these lines of work is that any apparent improvement in the condition number does not automatically translate into a quantum speedup, because the quantum realization of the preconditioned operator can introduce normalization or synthesis overheads that offset the condition-number gain.

This phenomenon has recently been formalized as an obstruction in the oracle model. Lapworth and S{\"u}nderhauf~\cite{lapworth2025preconditioned} demonstrate that constructing a block-encoding of $PA$ by multiplying independent block-encodings of $A$ and $P$ can eliminate the benefit of preconditioning entirely, since the subnormalization of the product equals the product of the subnormalizations. Pechan et al.~\cite{pechan2025block} reach a comparable conclusion in a PDE setting for the three-dimensional heterogeneous Poisson operator, and advocate for a direct block-encoding of the preconditioned operator in place of a composed one. Collectively, these results establish two complementary conclusions: first, a classical reduction in the condition number is insufficient to produce a quantum advantage; second, the oracle model through which the preconditioned operator is accessed is also a decisive factor.

This access-model dependence is also consistent with recent work on quantum data access and structured block-encoding. Structured matrix input models show that both the cost and the normalization of a block-encoding may depend sensitively on the representation used for the matrix, including sparsity patterns, repeated values, and arithmetic structure~\cite{sunderhauf2024blockencoding,camps2024explicit,zhang2024circuit}. Thus, for preconditioned QLS solvers, it is important to distinguish separate access to the factors $A$ and $P$ from direct access to the classically formed product.

These observations motivate a closer examination of the interplay between the structural form of the system matrix and the oracle model through which the preconditioned operator is realized. In a wide range of applications, including discretized PDEs~\cite{childs2021high,sato2024hamiltonian,sturm2025efficient}, lattice models in condensed matter physics~\cite{lloyd1996universal,somma2002simulating}, and electronic structure Hamiltonians in quantum chemistry~\cite{jordan1928uber,mcardle2020quantum,cao2019quantum,babbush2019quantum}, the system matrix $A$ is Hermitian and admits a Pauli expansion
\begin{equation}
  \label{eq: A}
  A \;=\; \sum_{\ell=1}^{L} a_\ell\, A_\ell,
  \qquad A_\ell\in\mathcal{V}_n,\ a_\ell\in\mathbb{R},
\end{equation}
where $\mathcal{V}_n$ denotes the set of $n$-qubit Pauli words~\cite{nielsen2010quantum} and $n:=\log_2 N$, $1 \le L \le 4^n$.

Accordingly, for matrices of the form~\eqref{eq: A}, QLS algorithms may be viewed through two widely used access models.
The first is the coherent block-encoding model, in which one assumes or constructs a unitary block-encoding of the system matrix and then applies QLS procedures based on qubitization, quantum singular value transformation, or related matrix-inversion techniques~\cite{gilyen2019quantum,martyn2021grand}.
Pauli expansions provide a natural input format for this model through LCU- and qubitization-based block-encoding constructions, and explicit circuit constructions for block-encoding linear combinations of Pauli strings have been studied in recent work~\cite{berry2015simulating,gilyen2019quantum,low2019hamiltonian,zhang2024circuit,schillo2026block}.
The second is the randomized Pauli access model of Wang et al.~\cite{wang2024qubit}.
Randomized QLS solvers provide a complementary paradigm for quantum linear systems by estimating classical information about $A^{-1}b$, rather than coherently preparing the full solution state.
This makes them relevant for tasks where only selected entries, overlaps, or scalar observables of the solution are required, and where access to a Pauli expansion is more natural than access to a block-encoding oracle.
Several closely related developments clarify the scope of this model.
Low-ancilla implementations of linear combinations of unitaries, including single-ancilla and ancilla-free variants with applications to quantum linear systems, were developed independently and contemporaneously in~\cite{chakraborty2024lcu}.
More recent work pursues the same broad goal of reducing reliance on fully coherent block-encoding oracles.
Near-optimal QSVT without block-encodings and with minimal ancilla overhead was developed in~\cite{chakraborty2025qsvtwithout}, while randomized QSVT replaces block-encodings by randomized sampling and yields randomized algorithms for quantum linear systems~\cite{wang2025randomizedqsvt}.
These results place randomized Pauli-access algorithms within a broader family of oracle-light and low-ancilla approaches.
For such randomized methods, reductions in circuit depth must be assessed in conjunction with the associated sampling cost, as emphasized by recent end-to-end resource analyses of randomized QLS solvers~\cite{hariprakash2025randomized}.

Motivated by the classical technique of sparse approximate inverses (SPAI)~\cite{benzi2002preconditioning,grote1997parallel}, we further require that the preconditioner $P$ respect the same structural form, and accordingly take
\begin{equation}
  \label{eq: preconditioner}
  P \;=\; \sum_{m=1}^{M} \beta_m\, P_m,
  \qquad P_m\in\mathcal{V}_n,\ \beta_m\in\mathbb{R}.
\end{equation}
The representation~\eqref{eq: preconditioner} furnishes a Pauli-basis ansatz for the preconditioner.  When $M\ll 4^n$, it constitutes a sparse Pauli ansatz; the numerical witness in Section~\ref{sec: randomized numerical experiment}, however, employs the full diagonal family $\{I,Z\}^{\otimes n}$ and is therefore not advanced as a scalable sparse construction.  Algebraically, whenever $A$ admits the expansion~\eqref{eq: A}, the preconditioned operators $PA$ and $PAP^{\dagger}$ likewise admit Pauli expansions after regrouping identical Pauli words. The cost of forming, storing, block-encoding, or sampling from this regrouped expansion is part of the access model and is accounted for separately in the corresponding constructions. Moreover, the coefficient vector $\beta=(\beta_m)_{m=1}^{M}$ furnishes a finite-dimensional and interpretable design space, within which the preconditioner may be tuned against tractable complexity parameters, namely, the block-encoding factors in the block-encoding model and the Pauli coefficient weight in the randomized Pauli access model.
Our contributions are as follows.

\begin{itemize}
    \item We develop a Pauli-structured framework for preconditioned quantum linear system solvers.
    In a unified oracle-model formulation, we formalize the known normalization obstruction associated with composing separate block-encodings of the system matrix and the preconditioner.
    Motivated by this obstruction, and starting from Pauli expansions of the system matrix and the preconditioner, we show how the left- and symmetrically preconditioned systems can be represented directly as regrouped Pauli expansions.
    The regrouping procedure collects identical Pauli words and yields explicit bounds on the resulting list sizes and coefficient weights.
    This provides a direct access model for the preconditioned system, in which the relevant normalization and implementation costs are determined by the regrouped Pauli representation itself rather than by separate access to the individual factors.

    \item We analyze the implications of this framework for the randomized QLS solver.
    In this setting, the dominant structural depth parameter is controlled by two competing quantities: the Pauli coefficient weight of the input matrix and its smallest singular value.
    We show that Pauli-structured preconditioners can improve the randomized QLS per-sample depth proxy only when the improvement in the stability of the preconditioned system is not offset by the increase in the regrouped Pauli coefficient weight.
    This gives a concrete design principle for constructing effective Pauli-structured preconditioners in the randomized Pauli access model.

    \item We provide a numerical witness for this mechanism.
    On a synthetic Pauli-structured benchmark family, we construct a least-squares approximate inverse within a diagonal Pauli ansatz and evaluate the improvement obtained from the resulting Pauli-structured preconditioner for both direct block-encoding diagnostics and randomized QLS solvers.
    Across all tested sizes, the randomized QLS per-sample depth proxy is reduced relative to the unpreconditioned baseline, and the norm-aware effective QLS condition parameter for the directly encoded symmetric product is also reduced.
    These experiments demonstrate that Pauli regrouping after preconditioning can produce proxy-level improvements in both settings.

\end{itemize}

\subsection{Organization}

Section~\ref{sec: standard qls preconditioning} formulates the limitation of separate block-encoding composition and uses it to motivate the constructive settings considered later.
Section~\ref{sec: pauli structured data} turns to direct block-encoding of classically formed preconditioned operators from Pauli-structured data.
Section~\ref{sec: randomized qls preconditioning section} studies Pauli-structured preconditioners in the randomized Pauli access model and derives the corresponding complexity bounds.
Section~\ref{sec: randomized numerical experiment} provides a numerical illustration of the predicted improvement mechanism.

\subsection{Notation}
\label{sec: notation}

We use \(\|\cdot\|\) for the spectral norm of a matrix or the $\ell_2$ norm of a vector, depending on context. For an invertible matrix \(H\), we write \(\sigma_{\min}(H)\) for its smallest singular value and \(\kappa(H)=\|H\|/\sigma_{\min}(H)\) for its condition number. The notation \(\widetilde{\mathcal{O}}(\cdot)\) suppresses polylogarithmic factors in the displayed parameters. \(\mathcal{V}_n\) denotes the set of \(n\)-qubit Pauli words. If a matrix $H$ admits a Pauli expansion, then after collecting identical Pauli words, we write $H= \sum_{j=1}^J h_j H_j$, where $h_j \in \mathbb{C}, H_j \in \mathcal{V}_n$, and $1 \leq J \leq 4^n$. Its Pauli coefficient weight is defined as $w(H):=\sum_{j=1}^J\left|h_j\right|$. If $H$ is Hermitian, then the coefficients in the regrouped expansion over the Hermitian Pauli basis may be taken to be real.

\section{Limitations of separate block-encoding composition}
\label{sec: standard qls preconditioning}

We first isolate a limitation that does not rely on any Pauli structure.
This section delineates what can and cannot be achieved when the preconditioned operator is
realized by composing separate block-encodings of the system matrix and the
preconditioner.

Recent work has demonstrated that, in QLS solvers based on block-encodings, the effect of preconditioning depends crucially on how the
preconditioned operator is implemented.
Lapworth and S{\"u}nderhauf~\cite{lapworth2025preconditioned} compare
separate block-encoding of $A$ and $P$ followed by quantum multiplication against direct
block-encoding of the classically formed product, and observe that the former approach can forfeit its
advantage owing to subnormalization overhead.
Pechan, Golden, and O'Malley~\cite{pechan2025block} identify a critical
limitation of separately block-encoding the system matrix and the
preconditioner in a PDE setting.
The theorem below isolates this obstruction in the separate block-encoding composition model.

As in the standard QLSP formulation introduced below, we assume throughout this section that
$A \in \mathbb{C}^{N\times N}$ is Hermitian, positive semidefinite, and
normalized so that $\|A\|=1$.
These standing assumptions are consistent with the standard Hermitian formulation of the
QLSP; the obstruction proved below at the oracle-model level, however, depends only on the
smallest singular value and not on positivity.
No structural assumption on $P$ is needed here beyond the availability of a
block-encoding.

\subsection{Quantum linear system problem and its complexity}
\label{sec: qls problem and complexity}
A quantum linear system problem (QLSP) is formulated as follows.

\begin{problem}[QLSP]
  \label{problem: QLSP}
  Following the standard formulation in~\cite{morales2024quantum}, we are given an invertible matrix $A \in \mathbb{C}^{N \times N}$ and a vector $b \in \mathbb{C}^N$. Assume without loss of generality that (i) $A$ is Hermitian, positive semidefinite, and its singular values lie in $[\sigma_{\min}(A),1]$, and (ii) $\|b\|=1$. Let $x=A^{-1} b$ be the solution to the linear system. Denote the associated quantum states of $b$ and $x$ by
\[
|b\rangle=\sum_{i=1}^N b_i|i\rangle, \qquad
|x\rangle=\frac{\sum_{i=1}^N x_i|i\rangle}{\| \sum_{i=1}^N x_i|i\rangle \|}.
\]
Assuming access to $A$ and a state preparation oracle $U_b$ such that $U_b \lvert 0 \rangle = \lvert b \rangle$, the goal is to return a state $|\tilde{x}\rangle$ satisfying $\||\tilde{x}\rangle-|x\rangle \|< \varepsilon$ for a prescribed error tolerance $\varepsilon>0$.
\end{problem}

\begin{remark}
  \label{remark: qlsp psd wlog}
The assumption \(A\succeq 0\) in Problem~\ref{problem: QLSP} is an expositional simplification. For an invertible Hermitian matrix \(A\), positivity is not essential. Indeed, direct-inversion and QSVT-based QLS solvers may approximate \(1/x\) directly on the two-sided spectral domain \([-1,-\sigma_{\min}(A)]\cup[\sigma_{\min}(A),1]\). We therefore use the PSD formulation in this section only as a clean baseline; later sections may work with invertible Hermitian matrices that are not positive semidefinite.

  For a
  non-Hermitian system, one may use the standard Hermitian dilation
  $
    \begin{pmatrix}
      0 & A \\
      A^\dagger & 0
    \end{pmatrix},
  $
  which preserves the singular values of \(A\).  With this block convention, the
  embedded system corresponding to \(Ax=b\) is
  \[
    \begin{pmatrix}
      0 & A \\
      A^\dagger & 0
    \end{pmatrix}
    \begin{pmatrix}
      0\\x
    \end{pmatrix}
    =
    \begin{pmatrix}
      b\\0
    \end{pmatrix},
    \qquad x=A^{-1}b.
  \]
  Thus the desired solution appears in the lower block.  If the right-hand side
  is instead placed in the lower block, the upper block solves the corresponding
  adjoint system, so the block placement must be fixed consistently.
\end{remark}

Assume first that access to $A$ is given by a $(1,a_A,\varepsilon_A)$
block-encoding $U_A$.
We model the cost of the QLS solver purely by the number of oracle calls to
$U_A$ and $U_b$.
Let $\kappa_A = \frac{1}{\sigma_{\min}(A)}$.
Then the query complexity of a general QLS solver for
Problem~\ref{problem: QLSP} may be written as
\begin{equation}
  \label{eq: QLS complexity}
  \mathrm{T}_A\!\left(\kappa_A,\frac{1}{\varepsilon_A},\frac{1}{\varepsilon}\right), \quad
  \mathrm{T}_b\!\left(\kappa_A,\frac{1}{\varepsilon_A},\frac{1}{\varepsilon}\right)
\end{equation}
for $U_A$ and $U_b$, respectively. $\mathrm{T}_A$ and $\mathrm{T}_b$ are monotonically increasing in their arguments.

If instead access to $A$ is given by an $(\alpha_A,a_A,\varepsilon_A)$
block-encoding with $\alpha_A \ge 1$, then the same unitary may be viewed as a
$(1,a_A,\frac{\varepsilon_A}{\alpha_A})$ block-encoding of the rescaled matrix
$A' := \frac{A}{\alpha_A}$.
The singular values of $A'$ lie in
$[\frac{\sigma_{\min}(A)}{\alpha_A},\frac{1}{\alpha_A}] \subseteq [\frac{\sigma_{\min}(A)}{\alpha_A},1]$.
Therefore the query complexity of the
linear system solver becomes
\begin{equation}
  \label{eq: QLS complexity 2}
  \mathrm{T}_A\!\left(\alpha_A \kappa_A,\frac{\alpha_A}{\varepsilon_A},\frac{1}{\varepsilon}\right), \quad
  \mathrm{T}_b\!\left(\alpha_A \kappa_A,\frac{\alpha_A}{\varepsilon_A},\frac{1}{\varepsilon}\right)
\end{equation}
for $U_A$ and $U_b$, respectively.
For notational convenience, for any invertible matrix $A$ accessed through an $(\alpha_A, a_A, \varepsilon_A)$ block-encoding, we define the effective QLS condition parameter associated with the block-encoding normalization as
\begin{equation}
  \label{eq: effective qls condition parameter}
  \kappa_{\mathrm{eff}}(A;\alpha_A)
  :=
  \frac{\alpha_A}{\sigma_{\min}(A)}.
\end{equation}
It coincides with the usual condition number of \(A\), or equivalently of \(\frac{A}{\alpha_A}\), only when the normalization is tight, namely when \(\alpha_A=\|A\|\).

\subsection{Preconditioned systems from separate block-encodings}
\label{sec: preconditioned linear systems}
In this subsection, we consider a general QLS solver given an
$(\alpha_A,a_A,\varepsilon_A)$ block-encoding $U_A$ of the system matrix $A$,
with $\|A\| \le 1$ and $\alpha_A \ge 1$, together with a state preparation
oracle $U_b$ such that $U_b\lvert 0\rangle=\lvert b\rangle$.
Let $\varepsilon>0$ be the target accuracy, and use the notation from Section~\ref{sec: qls problem and complexity}.

Assume that we are given an $(\alpha_P,a_P,\varepsilon_P)$ block-encoding $U_P$ of
the preconditioner $P$.
In the limitation result below, the normalization factors are used as certified
upper bounds: \(\alpha_A\ge\|A\|\) and \(\alpha_P\ge\|P\|\).  The first condition
holds in the present normalization because \(\|A\|\le1\) and \(\alpha_A\ge1\).
The second condition is automatic for exact block-encodings, but for approximate
block-encodings it should be imposed as an independent normalization certificate.
If one only knows an approximate block-encoding in the sense that its encoded
block \(B_P\) satisfies \(\|P-\alpha_P B_P\|\le\varepsilon_P\), then one has only
\(\|P\|\le\alpha_P+\varepsilon_P\), and the proof below gives the corresponding
slackened bound with this replacement.
For the transformed linear system to be equivalent to the original one, one
normally assumes that $P$ is invertible; however, the lower bound at the oracle-model level
proved below does not require invertibility of $P$.
We treat left and symmetric preconditioning in a unified way.

For left preconditioning, the transformed system is
\[
PAx = Pb.
\]
Even when both $P$ and $A$ are Hermitian, the product $PA$ need not be
Hermitian unless $P$ and $A$ commute.
Accordingly, whenever a standard Hermitian QLS solver is invoked for the
left preconditioned system, one should either restrict to a Hermitian $PA$ or
pass to the standard Hermitian embedding of $PA$.
This distinction does not affect the singular value inequalities underlying the
limitation result.

For symmetric preconditioning, the transformed system is
\[
PAP^\dagger y = Pb,
\]
with the original solution recovered as $x=P^\dagger y$.
Since $A$ is Hermitian, $PAP^\dagger$ is Hermitian for every $P$.

\paragraph{Block-encoding of the preconditioned operator.}
For left preconditioning, the product unitary $U_PU_A$ is an
$(\alpha_L,a_L,\varepsilon_L)$ block-encoding of $PA$, where
\[
\alpha_L=\alpha_P\alpha_A, \qquad a_L=a_P+a_A.
\]
The block-encoding error follows from the standard multiplication lemma for
block-encodings, for example Lemma~48 of~\cite{gilyen2019quantum}.
Applying that lemma to an $(\alpha_P,a_P,\varepsilon_P)$ block-encoding of $P$
and an $(\alpha_A,a_A,\varepsilon_A)$ block-encoding of $A$ gives
\[
\varepsilon_L
= \alpha_P\varepsilon_A + \alpha_A\varepsilon_P + \varepsilon_P\varepsilon_A.
\]
Equivalently, after dividing by $\alpha_L=\alpha_P\alpha_A$, this gives a
$(1,a_L,\frac{\varepsilon_L}{\alpha_L})$ block-encoding of the normalized matrix
$
A_L := \frac{1}{\alpha_L}PA,
$
with normalized error
\[
\frac{\varepsilon_L}{\alpha_L}
= \frac{\varepsilon_A}{\alpha_A}
  + \frac{\varepsilon_P}{\alpha_P}
  + \frac{\varepsilon_P\varepsilon_A}{\alpha_P\alpha_A}.
\]

For symmetric preconditioning, first regard $U_PU_A$ as an
$(\alpha_L,a_L,\varepsilon_L)$ block-encoding of $PA$, and then multiply by
$U_P^\dagger$.
The resulting unitary
$
U_{PAP^\dagger}:=U_PU_AU_P^\dagger
$
is an $(\alpha_S,a_S,\varepsilon_S)$ block-encoding of $PAP^\dagger$, where
\[
\alpha_S=\alpha_P^2\alpha_A, \qquad
 a_S=a_A+2a_P,
\]
and another application of the same multiplication lemma yields
\[
\varepsilon_S
= \alpha_P\varepsilon_L + \alpha_L\varepsilon_P + \varepsilon_L\varepsilon_P.
\]
Equivalently, this gives a $(1,a_S,\frac{\varepsilon_S}{\alpha_S})$ block-encoding of
the normalized matrix
$
A_S := \frac{1}{\alpha_S}PAP^\dagger,
$
with normalized error
\[
\frac{\varepsilon_S}{\alpha_S}
= \frac{\varepsilon_L}{\alpha_L}
  + \frac{\varepsilon_P}{\alpha_P}
  + \frac{\varepsilon_L\varepsilon_P}{\alpha_P\alpha_L}.
\]

If the solver returns a normalized state \(|\bar y\rangle=\frac{y}{\|y\|}\) proportional
to the solution \(y\) of the symmetric preconditioned system, then to obtain a
state proportional to \(x=P^\dagger y\) we apply \(U_P^\dagger\) and amplify the
corresponding block.
This yields a state preparation oracle for
$
|x\rangle := \frac{P^\dagger y}{\|P^\dagger y\|}
$
with query overhead
\[
\mathrm{T}_x
=
\mathcal{O}\!\left(
  \frac{\alpha_P}{\|P^\dagger|\bar y\rangle\|}
  \log\frac{1}{\varepsilon}
\right).
\]

Thus, in both cases, the relevant effective QLS condition parameter is determined by
the smallest singular value of the normalized preconditioned operator.
Let
\[
\kappa_L := \kappa_{\mathrm{eff}}(PA;\alpha_L), \qquad
\kappa_S := \kappa_{\mathrm{eff}}(PAP^\dagger;\alpha_S).
\]
If the corresponding preconditioned operator is singular, we set the associated
effective QLS condition parameter to \(+\infty\).  With this convention, by
Theorem~\ref{thm:no-go-preconditioning}, neither $\kappa_L$ nor $\kappa_S$
can be smaller than the baseline value $\alpha_A\kappa_A$.
This is a statement at the oracle-model level about the effective QLS condition parameter alone.
The solver-dependent consequences are recorded separately below.

\paragraph{Preparation of the preconditioned right-hand side.}
Given $U_b\,|0\rangle = |b\rangle$, in both settings we need a state
proportional to $Pb$.
This can be obtained by applying $U_P$ to $|0\rangle|b\rangle$ and
postselecting, or coherently amplifying, the ancilla outcome corresponding to
the encoded block.
Using oblivious amplitude amplification~\cite{gilyen2019quantum}, one can implement a state preparation
oracle for
$
\frac{Pb}{\|Pb\|}
$
with query overhead
$
  \mathrm{T}_r = \mathcal{O}\!\left(\frac{\alpha_P}{\|Pb\|}\,\log\frac{1}{\varepsilon}\right).
$
Therefore, the query complexity with respect to $U_b$ (or its adjoint) is
\[
  \mathrm{T}_r\mathrm{T}_b\!\left(\kappa_L, \frac{\alpha_L}{\varepsilon_L}, \frac{1}{\varepsilon}\right), \quad
  \mathrm{T}_r \mathrm{T}_b\!\left(\kappa_S, \frac{\alpha_S}{\varepsilon_S}, \frac{1}{\varepsilon}\right)
\]
for left preconditioning and symmetric preconditioning, respectively.

The key point is that, in the separate block-encoding composition model, improving the classical
conditioning of $PA$ or $PAP^\dagger$ is not by itself sufficient.
The relevant QLS parameter also includes the normalization inherited from the
available block-encodings.

\begin{theorem}
  \label{thm:no-go-preconditioning}
  (\textbf{Limitation of separate block-encoding composition})
  Assume access to the system matrix $A$ via an $(\alpha_A,a_A,\varepsilon_A)$ block-encoding $U_A$, with $\|A\|\le 1$, $\alpha_A\ge1$, and $\kappa_A:=\frac{1}{\sigma_{\min}(A)}$. Let $P$ be any preconditioner for which we have an $(\alpha_P,a_P,\varepsilon_P)$ block-encoding $U_P$, and suppose that the normalization of the preconditioner is certified in the sense that $\alpha_P\ge\|P\|$.
This includes exact block-encodings and approximate block-encodings accompanied by an independently certified normalization bound.
If a preconditioned operator is singular, its effective QLS condition parameter is
understood as \(+\infty\).  Then the effective QLS condition parameters satisfy
$\kappa_L \ge \alpha_A\kappa_A$ and $\kappa_S \ge \alpha_A\kappa_A$.
\end{theorem}

\begin{proof}
  If $PA$ or $PAP^\dagger$ is singular, the corresponding inequality is true by
  the convention above.  We therefore prove the claims in the nonsingular cases.
  By the certified normalization assumption, \(\alpha_P\ge \|P\|\).
  For left preconditioning, use $\sigma_{\min}(PA) \le \|P\|\,\sigma_{\min}(A)$ to obtain
  \[
    \kappa_L=\frac{\alpha_P\alpha_A}{\sigma_{\min}(PA)}
    \ge \frac{\alpha_P\alpha_A}{\|P\|\sigma_{\min}(A)}
    \ge \frac{\alpha_A}{\sigma_{\min}(A)}
    = \alpha_A\kappa_A.
  \]
  For symmetric preconditioning, it follows that
  \[
    \kappa_S
    = \frac{\alpha_P^2\alpha_A}{\sigma_{\min}(PAP^\dagger)}
    \ge \frac{\alpha_P^2\alpha_A}{\|P\|^2\sigma_{\min}(A)}
    \ge \frac{\alpha_A}{\sigma_{\min}(A)}
    = \alpha_A\kappa_A.
  \]
\end{proof}

\begin{corollary}
\label{cor: no-go-condition-number-contribution}
Consider any QLS solver whose oracle query complexity is monotone nondecreasing in the effective QLS condition parameter. Then, within the separate block-encoding composition model, preconditioning cannot improve the contribution to the query complexity coming from the effective QLS condition parameter.
\end{corollary}

\begin{proof}
By Theorem~\ref{thm:no-go-preconditioning}, one has $\kappa_L \ge \alpha_A\kappa_A$ and $\kappa_S \ge \alpha_A\kappa_A$. Therefore, for any solver whose query complexity is monotone nondecreasing in the effective QLS condition parameter, replacing $A$ by $PA$ or $PAP^\dagger$ in the separate block-encoding composition model cannot reduce the part of the complexity that depends on this parameter.
\end{proof}

For the particular abstract complexity model in \eqref{eq: QLS complexity}--\eqref{eq: QLS complexity 2}, this yields the
solver-dependent inequalities
\begin{equation}
\begin{aligned}
\mathrm{T}_A\!\left(\kappa_L,\frac{\alpha_L}{\varepsilon_L},\frac{1}{\varepsilon}\right)
&\ge
\mathrm{T}_A\!\left(\alpha_A\kappa_A,\frac{\alpha_L}{\varepsilon_L},\frac{1}{\varepsilon}\right), \\
\mathrm{T}_r\,\mathrm{T}_b\!\left(\kappa_L,\frac{\alpha_L}{\varepsilon_L},\frac{1}{\varepsilon}\right)
&\ge
\mathrm{T}_r\,\mathrm{T}_b\!\left(\alpha_A\kappa_A,\frac{\alpha_L}{\varepsilon_L},\frac{1}{\varepsilon}\right),
\end{aligned}
\end{equation}
for left preconditioning, with analogous inequalities for symmetric preconditioning obtained by replacing $(\kappa_L,\alpha_L,\varepsilon_L)$ by $(\kappa_S,\alpha_S,\varepsilon_S)$. These inequalities isolate only the effective QLS condition parameter contribution. The full query complexity also depends on the precision parameters $\frac{\alpha_L}{\varepsilon_L}$, $\frac{\alpha_S}{\varepsilon_S}$, and on the right-hand side preparation and recovery overheads $\mathrm{T}_r$ and $\mathrm{T}_x$.

\begin{remark}
  \label{remark: scope-no-go}
Theorem~\ref{thm:no-go-preconditioning} is consistent with the negative message in \cite{lapworth2025preconditioned} for quantum multiplication of separately encoded factors, and also with the limitations of separate block-encoding composition emphasized in \cite{pechan2025block} for separately block-encoding a system matrix and a preconditioner. In all of these cases, the normalization or subnormalization overhead can offset the reduction in the classical condition number at the level of the effective QLS condition parameter.

At the same time, Theorem~\ref{thm:no-go-preconditioning} does not rule out a different regime in which the classically formed product $PA$ or $PAP^\dagger$ is block-encoded directly as a new matrix. Therefore, the theorem should be interpreted as a limitation of composition of separate block-encodings, rather than as a general impossibility result for all forms of quantum preconditioning. Compared with~\cite{lapworth2025preconditioned}, Theorem~\ref{thm:no-go-preconditioning} provides a self-contained oracle-model lower bound valid for any preconditioner $P$ (without commutativity or positivity assumptions) and treats left and symmetric preconditioning in a unified way via the certified normalization condition $\alpha_P\ge\|P\|$.
\end{remark}

\section{Pauli-structured preconditioners with direct block-encoding}

\label{sec: pauli structured data}

The limitation result in Section~\ref{sec: standard qls preconditioning} applies to a
specific access model: the preconditioned operator is realized by multiplying
separate block-encodings of \(P\) and \(A\).  We now consider a different
access model.  The product $Q\in\{PA, PAP^\dagger\}$ is first formed classically from the Pauli descriptions of \(A\) and \(P\), and
the resulting operator \(Q\) is then block-encoded directly.  This section makes precise when the direct construction avoids the
normalization penalty of separate block-encoding composition and what additional
costs must enter the comparison.

Throughout this section, assume that \(A\) has the Pauli expansion~\eqref{eq: A}, and that the preconditioner has the Pauli form~\eqref{eq: preconditioner}.  We write
\[
  w(A):=\sum_{\ell=1}^{L}|a_\ell|,
  \qquad
  w(P):=\sum_{m=1}^{M}|\beta_m|,
\]
for the Pauli coefficient weights as defined in Section~\ref{sec: notation}.

\paragraph{Admissible preconditioned systems.}
For the left preconditioned system
\[PAx=Pb,\]
equivalence with the original
linear system requires \(P\) to be invertible.  If a Hermitian QLS solver is used
directly on \(PA\), then \(PA\) must also be Hermitian.  Since \(A\) and \(P\)
are Hermitian, this is equivalent to
$
  [P,A]=0.
$
If the solver formulation additionally assumes positive definiteness, it is
enough to assume \(A\succ0\), \(P\succ0\), and \([P,A]=0\), in which case
\(PA\succ0\).  If these assumptions are not imposed, one should consider the following symmetric preconditioned system
\[
  PAP^\dagger y=Pb,
\]
with the substitution \(x=P^\dagger y\).
Equivalence of the symmetric system with the original system also requires
\(P\) to be invertible; equivalently, since \(A\) is invertible, it is enough to
assume that \(PAP^\dagger\) is invertible.  The operator \(PAP^\dagger\) is
Hermitian whenever \(A\) is Hermitian.  The following proposition gives an
immediate implication for the Pauli coefficient weight.

\begin{proposition}
  \label{prop: pauli-closure-direct-preconditioning}
  Let \(A\) and \(P\) be given by~\eqref{eq: A} and~\eqref{eq: preconditioner}.
  Then the following statements hold.

  \begin{enumerate}
    \item If \(PA\) is Hermitian, then after regrouping identical Pauli words,
    $
      PA=\sum_{j=1}^{J_L}\mu_j^{(L)}Q_j^{(L)},
      J_L\le LM,
    $
    with \(Q_j^{(L)}\in\mathcal{V}_n\) and \(\mu_j^{(L)}\in\mathbb{R}\).  Its
    Pauli coefficient weight satisfies
    $
      w(PA)\le w(P) w(A).
    $

    \item The symmetric product admits a Pauli expansion
    $
      PAP^\dagger=\sum_{j=1}^{J_S}\mu_j^{(S)}Q_j^{(S)},
      J_S\le LM^2,
    $
    with \(Q_j^{(S)}\in\mathcal{V}_n\) and \(\mu_j^{(S)}\in\mathbb{R}\).  Its
    Pauli coefficient weight satisfies
    $
      w(PAP^\dagger)\le w(P)^2 w(A).
    $
  \end{enumerate}
\end{proposition}

\begin{proof}
  The product of Pauli words is a phase times a Pauli word.  Thus every term
  \(P_mA_\ell\) and every term \(P_mA_\ell P_{m'}^\dagger\) can be rewritten as
  a scalar multiple of a Pauli word.  Regrouping identical Pauli words gives at
  most \(LM\) and \(LM^2\) distinct terms, respectively.  The coefficient
  bounds follow from the triangle inequality:
  $
    w(PA)
    \le
    \sum_{m,\ell}|\beta_m|\,|a_\ell|
    =
    w(P) w(A),
  $
  and
  $
    w(PAP^\dagger)
    \le
    \sum_{m,\ell,m'}|\beta_m|\,|a_\ell|\,|\beta_{m'}|
    =
    w(P)^2 w(A).
  $
  Finally, a Hermitian matrix has real coefficients in the Hermitian Pauli
  basis, so the regrouped coefficients are real in the two Hermitian cases
  stated above.
\end{proof}

The proposition shows that the classically formed operators remain in the
Pauli input model.  We can therefore analyze both cases uniformly.  Let
\begin{equation}
  \label{eq: generic-direct-pauli-expansion}
  Q=\sum_{j=1}^{J}\mu_j Q_j,
  \qquad
  Q_j\in\mathcal{V}_n,
  \qquad
  \mu_j\in\mathbb{R},
\end{equation}
where \(Q=PA\) in the admissible left preconditioned case and
\(Q=PAP^\dagger\) in the symmetric case, and the \(Q_j\) are the Pauli words appearing in the regrouped expansion.  Given Pauli access to the \(Q_j\), the operator \(Q\) admits a direct block-encoding through the standard linear combination of unitaries (LCU) framework~\cite{gilyen2019quantum,low2019hamiltonian,zhang2024circuit} (see Appendix~\ref{app: LCU implementation} for implementation details).  A more efficient block-encoding can be obtained by simulating \(e^{itQ}\) via a product formula~\cite{childs2021theory} and then applying the matrix logarithm construction of Appendix~\ref{app: hamiltonian simulation}.  The cost of the resulting block-encoding of \(Q\) is analyzed in the following proposition.

For an integer \(p\ge1\), define the corresponding commutator scaling \(\widetilde{\mu}_{\mathrm{comm}}^{(p)}(Q)\) by
\begin{equation}
  \label{eq: direct-Q-commutator-scaling}
  \sum_{j_1,\ldots,j_{p+1}=1}^{J}
  \left\|
  \left[
    i\mu_{j_{p+1}}Q_{j_{p+1}},
    \left[
      \cdots
      \left[
        i\mu_{j_2}Q_{j_2},
        i\mu_{j_1}Q_{j_1}
      \right]
      \cdots
    \right]
  \right]
  \right\|.
\end{equation}
This is the same quantity as in~\eqref{eq:alpha-comm-scaling}, specialized to
the Pauli expansion of \(Q\).

\begin{proposition}
  \label{prop: direct-log-exp-block-encoding-Q}
  (\textbf{Direct block-encoding of \(Q\)})
  Let \(Q\) be Hermitian, let \(\delta\in(0,\frac{1}{2}]\), and suppose that an upper
  bound \(\Lambda_Q\ge \|Q\|\) is available.  Choose
  $
    0<t\le \frac{\pi(1-\delta)}{2\Lambda_Q}.
  $
  Using the matrix logarithm construction from Appendix~\ref{app: hamiltonian simulation}
  together with a product formula of order \(p\) for \(e^{itQ}\), one obtains a
  $
    \left(\frac{\pi}{2t},\,2,\,\varepsilon_Q\right)
  $
  block-encoding of \(Q\).  Up to polylogarithmic factors in \(\frac{1}{\varepsilon_Q}\)
  and the polynomial dependence on \(\frac{1}{\delta^2}\) from the logarithm
  approximation, the number of Pauli exponentials is
  \begin{equation}
    \label{eq: direct-Q-implementation-cost}
    \begin{aligned}
    \mathrm{C}_Q(t,\varepsilon_Q)
    &=
    \widetilde{\mathcal{O}}\!\left(
      \Upsilon\,J\,
      \bigl(\widetilde{\mu}_{\mathrm{comm}}^{(p)}(Q)\bigr)^{\frac{1}{p}}
      t^{1+\frac{1}{p}}
      \delta^{-2(1+\frac{1}{p})}
      \varepsilon_Q^{-\frac{1}{p}}
    \right),
    \end{aligned}
  \end{equation}
  where \(\Upsilon\) is the number of exponentials in one segment of the chosen
  product formula.
\end{proposition}

\begin{proof}
  Since \(t\|Q\|\le t\Lambda_Q\le \frac{\pi(1-\delta)}{2}\), Proposition~\ref{prop: log implementation}
  applies to \(H=tQ\).  It gives a
  $
    \left(\frac{\pi}{2},\,2,\,\varepsilon_Q\right)
  $
  block-encoding of \(tQ\).  Equivalently, the same unitary is a
  $
    \left(\frac{\pi}{2t},\,2,\,\varepsilon_Q\right)
  $
  block-encoding of \(Q\).  The product formula cost for implementing \(e^{itQ}\)
  with commutator scaling gives~\eqref{eq: direct-Q-implementation-cost}.
\end{proof}

Taking the largest admissible time,
$
  t_Q:=\frac{\pi(1-\delta)}{2\Lambda_Q},
$
gives the direct block-encoding parameters
\begin{equation}
  \label{eq: direct-Q-block-encoding-parameters}
  (\alpha_Q,a_Q,\varepsilon_Q)
  =
  \left(
    \frac{\Lambda_Q}{1-\delta},
    2,
    \varepsilon_Q
  \right).
\end{equation}
Consequently, the effective QLS condition parameter of \(Q\) in the sense of
\eqref{eq: effective qls condition parameter} is
\begin{equation}
  \label{eq: direct-Q-effective-condition-number}
  \kappa_{\mathrm{eff}}(Q;\alpha_Q)
  =
  \frac{\alpha_Q}{\sigma_{\min}(Q)}
  =
  \frac{\Lambda_Q}{(1-\delta)\sigma_{\min}(Q)}.
\end{equation}
In the ideal norm-aware case \(\Lambda_Q=\|Q\|\), this reduces to
\(\frac{\kappa(Q)}{1-\delta}\).  The always available Pauli coefficient weight bound
\(\Lambda_Q=w(Q)\) follows from \(\|Q\|\le w(Q)\); if this is the only usable
bound, the direct effective QLS condition parameter becomes
\(\frac{w(Q)}{(1-\delta)\sigma_{\min}(Q)}\), rather than the norm-aware quantity
\(\frac{\kappa(Q)}{1-\delta}\).
The corresponding block-encoding precision parameter entering
\eqref{eq: QLS complexity} is
\begin{equation}
  \label{eq: direct-Q-precision-parameter}
  \frac{\alpha_Q}{\varepsilon_Q}
  =
  \frac{\Lambda_Q}{(1-\delta)\varepsilon_Q}.
\end{equation}

\begin{theorem}
  \label{thm: direct-pauli-preconditioning-complexity}
  (\textbf{QLS parameters for direct Pauli-structured preconditioners})
  Let \(Q\in\{PA,PAP^\dagger\}\) be an admissible Hermitian preconditioned
  operator as above, and suppose \(Q\) is invertible.  Suppose that an upper
  bound \(\Lambda_Q\ge\|Q\|\) is used and that \(Q\) is
  directly block-encoded using Proposition~\ref{prop: direct-log-exp-block-encoding-Q}
  with \(t=t_Q\).  Then a standard QLS solver applied to the system with
  coefficient matrix \(Q\) has matrix oracle query complexity
  \begin{equation}
	    \label{eq: direct-query-complexity-preconditioned-Q}
	    \mathrm{T}_A\!\left(
	      \frac{\Lambda_Q}{(1-\delta)\sigma_{\min}(Q)},
	      \frac{\Lambda_Q}{(1-\delta)\varepsilon_Q},
	      \frac{1}{\varepsilon}
	    \right).
  \end{equation}
  If a state proportional to \(Pb\) is prepared from \(U_b\) using an
  \((\alpha_P,a_P,\varepsilon_P)\) block-encoding of \(P\), then the additional
  right-hand side preparation overhead is
  \begin{equation}
    \label{eq: rhs-prep-direct-preconditioned-Q}
    \mathrm{T}_r
    =
    \mathcal{O}\!\left(
      \frac{\alpha_P}{\|Pb\|}
      \log\frac{1}{\varepsilon}
    \right).
  \end{equation}
  Hence the total query complexity with respect to the right-hand side
  preparation oracle is
  \begin{equation}
    \label{eq: direct-rhs-query-complexity-preconditioned-Q}
    \mathrm{T}_r\,
	    \mathrm{T}_b\!\left(
	      \frac{\Lambda_Q}{(1-\delta)\sigma_{\min}(Q)},
	      \frac{\Lambda_Q}{(1-\delta)\varepsilon_Q},
	      \frac{1}{\varepsilon}
	    \right).
  \end{equation}
  In the symmetric case, recovering a state proportional to
  \(x=P^\dagger y\) from a normalized solution state
  \(|\bar y\rangle=\frac{y}{\|y\|}\) has additional overhead
  \begin{equation}
    \label{eq: symmetric-recovery-overhead-direct}
    \mathrm{T}_x
    =
    \mathcal{O}\!\left(
      \frac{\alpha_P}{\|P^\dagger|\bar y\rangle\|}
      \log\frac{1}{\varepsilon}
    \right).
  \end{equation}
\end{theorem}

Theorem~\ref{thm: direct-pauli-preconditioning-complexity} isolates the main
difference from Theorem~\ref{thm:no-go-preconditioning}.  In the separate block-encoding composition
model, the effective normalization is inherited from the product of independent
block-encoding normalizations, namely \(\alpha_P\alpha_A\) for \(PA\) and
\(\alpha_P^2\alpha_A\) for \(PAP^\dagger\).  In the direct block-encoding construction, the
normalization is instead
\[
  \alpha_Q=\frac{\Lambda_Q}{1-\delta}.
\]
Therefore the effective QLS condition parameter contribution improves over an existing
\((\alpha_A,a_A,\varepsilon_A)\) block-encoding of \(A\) whenever
 $ \frac{\Lambda_Q}{(1-\delta)\sigma_{\min}(Q)}
  <
  \alpha_A\kappa_A$.
This is equivalent to
  $
  \frac{\Lambda_Q}{\sigma_{\min}(Q)}
  <
  (1-\delta)\alpha_A\kappa_A
$.
The block-encoding precision contribution is no worse than that of the baseline
oracle when
\begin{equation}
  \label{eq: direct-precision-benefit}
  \frac{\Lambda_Q}{(1-\delta)\varepsilon_Q}
  \lesssim
  \frac{\alpha_A}{\varepsilon_A}.
\end{equation}

Finally, an end-to-end primitive gate comparison must include the cost of
realizing one query to the directly encoded oracle.  If one call to the baseline
oracle \(U_A\) costs \(\mathrm{C}_A\) primitive operations, then the matrix oracle
part of the direct preconditioned method is favorable only if
\begin{equation}
  \label{eq: direct-end-to-end-benefit}
  \begin{aligned}
  &\mathrm{C}_Q(t_Q,\varepsilon_Q)\,
  \mathrm{T}_A\!\left(
    \frac{\Lambda_Q}{(1-\delta)\sigma_{\min}(Q)},
    \frac{\Lambda_Q}{(1-\delta)\varepsilon_Q},
    \frac{1}{\varepsilon}
  \right) \\
  &\quad <
  \mathrm{C}_A\,
  \mathrm{T}_A\!\left(
    \alpha_A\kappa_A,
    \frac{\alpha_A}{\varepsilon_A},
    \frac{1}{\varepsilon}
  \right).
  \end{aligned}
\end{equation}
Thus direct block-encoding removes the normalization obstruction from separate block-encoding composition, but
it does not by itself prove a full gate complexity advantage.  The advantage is
obtained only in the regime where the reduction in
\(\frac{\Lambda_Q}{\sigma_{\min}(Q)}\) and the chosen accuracy \(\varepsilon_Q\)
compensate for the Pauli expansion size \(J\), the
commutator scaling
\(\widetilde{\mu}_{\mathrm{comm}}^{(p)}(Q)\), and the right-hand side preparation
and recovery overheads.

\begin{remark}
  \label{remark: direct-preconditioning-vs-lapworth}
  The analysis above is consistent with the positive message of
  \cite{lapworth2025preconditioned}: once the classically formed operator
  \(Q\in\{PA,PAP^\dagger\}\) is block-encoded directly, the separate
  block-encoding obstruction in Theorem~\ref{thm:no-go-preconditioning} no
  longer applies.  The precise replacement is
  \[
	    \alpha_P\alpha_A,\ \alpha_P^2\alpha_A
	    \quad
	    \longrightarrow
	    \quad
	    \frac{\Lambda_Q}{1-\delta}.
	  \]
	  The price is that one must account for the direct realization cost
	  \(\mathrm{C}_Q(t_Q,\varepsilon_Q)\), rather than treating \(U_Q\) as a free
	  black box oracle.
	\end{remark}

\begin{remark}
  \label{remark: direct-classical-preprocessing}
  The direct construction is a structured access statement.  It assumes that the
  regrouped Pauli list for \(Q\) is available to the quantum implementation.
  Classically forming this list costs at most \(O(LM)\) Pauli multiplications
  before regrouping for \(Q=PA\), and at most \(O(LM^2)\) for
  \(Q=PAP^\dagger\).  Hashing or sorting is then needed to collect identical
  Pauli words, and the resulting list must be stored or sampled from.  These
  classical preprocessing and storage costs are not included in the matrix oracle
  query bounds above unless stated otherwise.
\end{remark}

\section{Pauli-structured preconditioners for the randomized QLS solver}
\label{sec: randomized qls preconditioning section}

For a matrix $A$ with Pauli expansion~\eqref{eq: A}, we now turn to the randomized QLS solver in the randomized Pauli access model of~\cite{wang2024qubit}.
In this setting, the preconditioner is constructed classically, and the quantum device processes only the resulting preconditioned instance.
This differs qualitatively from the separate block-encoding composition model of Section~\ref{sec: standard qls preconditioning}, and it is precisely in this regime that Pauli-structured preconditioners can yield a genuine improvement in the dominant structural proxy.

\subsection{Randomized quantum linear system problem and its complexity}
\label{sec: randomized QLS complexity}

We first recall the randomized quantum approach of Wang, McArdle, and Berta~\cite{wang2024qubit} for extracting classical information about $A^{-1}b$ without coherent oracle access, such as a block-encoding, for $A$.
The key distinction lies in the input model: instead of a unitary oracle $U_A$, one assumes a classical Pauli-basis description of $A$.

Let $n:=\log_2 N$, and suppose that $A$ admits the Pauli expansion~\eqref{eq: A} with $a_\ell\in\mathbb{R}$.
Define the Pauli coefficient weight by $w(A):=\sum_{\ell=1}^{L}|a_\ell|$.
In the randomized Pauli access model, one assumes classical access to the list $\{(a_{\ell},A_{\ell})\}_{\ell=1}^{L}$, or equivalently, the ability to sample indices $\ell$ with probabilities proportional to $|a_{\ell}|$ and to implement the corresponding Pauli words $A_{\ell}$ as quantum gates.

Following~\cite{wang2024qubit}, the randomized algorithm does not output the quantum state $|x\rangle\propto A^{-1}|b\rangle$.
Instead, it returns Monte Carlo estimates of scalar quantities derived from $A^{-1}|b\rangle$.
A representative end-to-end task, suitable for complexity comparison, is to estimate a single entry of the solution vector in the computational basis.
\begin{problem}
  \label{problem: randomized QLSP}
  (cf. Eq.~(23) in~\cite{wang2024qubit})
  Given an invertible Hermitian matrix $A\in\mathbb{C}^{N\times N}$ with Pauli description~\eqref{eq: A}, a vector $b\in\mathbb{C}^{N}$, an index $i\in[N]$, and a normalization constant $c\in\mathbb{C}$, the goal is to output an estimate of
  $
    \frac{1}{c}\,(A^{-1}b)_i
  $
  up to additive error $\varepsilon>0$ with constant success probability.
\end{problem}

\begin{remark}
  \label{remark: randomized QLS complexity}
  In the randomized framework, the dominant quantum resources are the number of circuit samples and the depth of each sampled circuit.
  For Hermitian matrices, the algorithmic space cost is $\log N+1$ qubits.
  For general non-Hermitian matrices, one can pass to a Hermitian embedding at the cost of one additional qubit, giving $\log N+2$ qubits~\cite{wang2024qubit}.

  Corollary~1 of~\cite{wang2024qubit} implies that for scalar functionals of the form $\frac{1}{q}\langle\psi|A^{-1}|b\rangle$, where $q>0$ is a freely chosen normalization parameter, one can obtain an additive-$\varepsilon$ estimate using
  \begin{equation}
    \label{eq: randomized sample complexity}
    C_{\mathrm{sample}}(A^{-1};\varepsilon,q)
    = \widetilde{\mathcal{O}}\!\left(\frac{1}{\sigma_{\min}(A)^{2}\,\varepsilon^{2}q^{2}}\right)
  \end{equation}
  circuit samples, where each sample has gate depth
  \begin{equation}
    \label{eq: randomized depth}
    C_{\mathrm{depth}}(A^{-1};\varepsilon)
    = \widetilde{\mathcal{O}}\!\left(\frac{w(A)^{2}}{\sigma_{\min}(A)^{2}}\,\log^{2}\!\left(\tfrac{1}{\varepsilon}\right)\right)
  \end{equation}
  up to the additional cost of preparing $|b\rangle$ and, when applicable, $|\psi\rangle$.
  Table~1 of~\cite{wang2024qubit} reports the corresponding end-to-end gate depth scaling for Problem~\ref{problem: randomized QLSP}.
  Recent resource analyses further show that the sampling overhead in randomized QLS solvers can be decisive for end-to-end performance~\cite{hariprakash2025randomized}.  Therefore, in the sequel we use the depth expression only to define a structural per-sample comparison, rather than as a claim of full end-to-end advantage.
\end{remark}

The crucial observation is that, in the randomized Pauli access model, the dominant structural parameter controlling gate depth is the Pauli coefficient weight $w(A)$ rather than a block-encoding normalization factor such as $\alpha_A$.
To avoid comparing ratios of asymptotic upper bounds directly, we use the following
scale-invariant structural proxy for the per-sample depth.  For an invertible
Pauli-expandable matrix \(H\), define
\begin{equation}
  \label{eq: rqls-depth-proxy}
  D_{\mathrm{RQLS}}(H)
  :=
  \left(
    \frac{w(H)}{\sigma_{\min}(H)}
  \right)^2.
\end{equation}
Thus the structural part of~\eqref{eq: randomized depth} is
\(D_{\mathrm{RQLS}}(A)\), up to polylogarithmic factors in the target accuracy
and state-preparation costs.
Accordingly, to study preconditioning in the randomized QLS solver, one should track how a Pauli-structured preconditioner changes both $w(A)$ and the relevant inverse norm parameters of the transformed system.

\subsection{Pauli-structured preconditioners for the randomized QLS solver}
\label{sec: randomized preconditioning}

We now study how a Pauli-structured preconditioner of the form~\eqref{eq: preconditioner} can be combined with the randomized QLS solver of~\cite{wang2024qubit}.
In contrast to standard QLS solvers based on block-encodings, the dominant parameters in the randomized Pauli access model are the Pauli coefficient weight $w(\cdot)$ and inverse norm factors such as $\|A^{-1}\|$.
We therefore track how these quantities change under preconditioning.

We focus on left preconditioning, namely $PAx=Pb$, which preserves the solution $x=A^{-1}b$ whenever $PA$ is invertible.
If \(PA\) is non-Hermitian, the linear system can be processed as in Remark~\ref{remark: qlsp psd wlog}.
Thus the desired solution is again obtained from the lower block, and the
normalized right-hand side used by the randomized solver is the normalized
version of \((Pb,0)^T\).

Suppose that $A$ admits the Pauli expansion~\eqref{eq: A} with real coefficients and that $P$ admits~\eqref{eq: preconditioner} with real coefficients.
Then
\begin{equation}
  \label{eq: pauli expansion PA}
  PA = \sum_{m=1}^{M}\sum_{\ell=1}^{L} \beta_m a_{\ell}\, P_m A_{\ell}.
\end{equation}
Since products of Pauli words are again Pauli words up to a phase, the preconditioned matrix $PA$ admits a Pauli expansion as well and therefore fits the same randomized Pauli access model after regrouping identical Pauli words.

We next record the basic quantitative bounds that follow from this observation.
The theorem isolates the two quantities that matter for the randomized QLS solver:
the singular value parameter and the regrouped Pauli coefficient weight.

\begin{theorem}[Randomized QLS]
  \label{thm: randomized qls preconditioning}
  If \(PA\) is invertible, then applying the
  randomized QLS solver in the randomized Pauli access model to \(PAx=Pb\),
  directly when \(PA\) is Hermitian and through the Hermitian embedding
  otherwise, gives the per-sample depth bound
  \begin{equation}
    \label{eq: randomized depth PA}
    C_{\mathrm{depth}}\bigl((PA)^{-1};\varepsilon\bigr)
    =
    \widetilde{\mathcal{O}}\!\left(
      \frac{w(PA)^{2}}{\sigma_{\min}(PA)^{2}}
      \log^{2}\!\left(\tfrac{1}{\varepsilon}\right)
    \right),
  \end{equation}
  up to state-preparation costs and the constant-factor overhead of the
  Hermitian embedding.  Equivalently, the structural part of the bound is
  \(D_{\mathrm{RQLS}}(PA)\).  Moreover, the sufficient condition
  \(\|I-PA\|\le\eta<1\) implies that \(PA\) is invertible,
  \(\sigma_{\min}(PA)\ge 1-\eta\), and
  \begin{equation}
    \label{eq: randomized depth PA neumann}
    C_{\mathrm{depth}}\bigl((PA)^{-1};\varepsilon\bigr)
    \le
    \widetilde{\mathcal{O}}\!\left(
      \frac{w(P)^2 w(A)^2}{(1-\eta)^2}
      \log^{2}\!\left(\tfrac{1}{\varepsilon}\right)
    \right).
  \end{equation}
  The same substitution rule applies to the sample-complexity bounds of
  \cite{wang2024qubit}: replace \(w(A)\) and \(\sigma_{\min}(A)\) by \(w(PA)\) and
  \(\sigma_{\min}(PA)\), respectively.  The normalization
  \(Pb/\|Pb\|\) scales the estimated functional by \(1/\|Pb\|\), which
  may affect the shot complexity through the normalization parameter \(q\) but
  not the per-sample depth proxy.
\end{theorem}

\begin{proof}
  The depth bound follows by applying the randomized complexity bounds of
  \cite{wang2024qubit} to \(PA\), or to its Hermitian embedding in the
  non-Hermitian case.  The embedding has the same nonzero singular values as
  \(PA\) and Pauli coefficient weight at most \(\sqrt{2}\,w(PA)\), which is
  absorbed into the \(\widetilde{\mathcal O}\) notation.  For the Neumann bound,
  let \(E:=I-PA\).  If \(\|E\|\le\eta<1\), then
  \[
    (PA)^{-1}=(I-E)^{-1}=\sum_{k=0}^{\infty}E^k,
    \qquad
    \|(PA)^{-1}\|\le \frac{1}{1-\eta}.
  \]
  Hence \(\sigma_{\min}(PA)\ge1-\eta\), and combining this with
  \(w(PA)\le w(P)w(A)\) gives~\eqref{eq: randomized depth PA neumann}.  The final
  statement is the corresponding substitution in the sample-complexity
  expression, together with the right-hand side normalization above.
\end{proof}

Theorem~\ref{thm: randomized qls preconditioning} makes the central trade-off in randomized QLS explicit.
The primary design objective is to reduce
\[
  \frac{w(PA)}{\sigma_{\min}(PA)},
\]
or equivalently to reduce \(w(PA)\) while maintaining a lower bound on
\(\sigma_{\min}(PA)\).  The product \(w(P)w(A)\) and the Neumann factor
\((1-\eta)^{-1}\) give a useful analyzable surrogate, but they are only upper
bounds.  The numerical examples below exploit cancellation after regrouping,
where \(w(PA)\) can be much smaller than the loose bound \(w(P)w(A)\).

\begin{remark}
  \label{remark: rqls preconditioning benefit}
  Theorem~\ref{thm: randomized qls preconditioning} shows that the randomized QLS solver is governed by two competing effects: an inverse stability factor and a Pauli coefficient weight factor.
  Rather than taking ratios of two \(\widetilde{\mathcal O}\) bounds, we compare the structural proxies defined in~\eqref{eq: rqls-depth-proxy}.  The preconditioned-to-unpreconditioned proxy ratio is
  \begin{equation}
    \label{eq: rqls depth proxy ratio}
    \frac{D_{\mathrm{RQLS}}(PA)}{D_{\mathrm{RQLS}}(A)}
    =
    \left(
      \frac{w(PA)/\sigma_{\min}(PA)}
           {w(A)/\sigma_{\min}(A)}
    \right)^2.
  \end{equation}
  Thus preconditioning reduces the dominant per-sample depth proxy whenever
  \[
    \frac{w(PA)}{\sigma_{\min}(PA)}
    <
    \frac{w(A)}{\sigma_{\min}(A)}.
  \]

  The same singular-value trade-off also appears in the shot complexity through
  the dependence on \(\frac{1}{\sigma_{\min}(PA)}\) and on the chosen normalization
  parameter \(q\) in~\eqref{eq: randomized sample complexity}.  Therefore
  \(D_{\mathrm{RQLS}}\) is a per-sample depth proxy, not a complete end-to-end
  sample-complexity or gate-complexity guarantee.
\end{remark}

\section{Numerical experiments for the randomized QLS solver}
\label{sec: randomized numerical experiment}

We give a concrete numerical illustration of the improvement mechanism described in Remark~\ref{remark: rqls preconditioning benefit}.
We consider the $n$-qubit family
\begin{equation}
  \label{eq: randomized numerical family}
  \begin{aligned}
  A
  &=
  c_0\,I^{\otimes n}
  +
  c_Z \sum_{i=1}^{n} Z_i \\
  &\quad +
  c_{ZZ} \sum_{i=1}^{n-1} Z_i Z_{i+1}
  +
  \varepsilon \sum_{i=1}^{n-1} (X_i X_{i+1} + Y_i Y_{i+1})
  \end{aligned}
\end{equation}
The resulting matrices are Hermitian and invertible but not positive semidefinite.
We fix the coefficients
\begin{equation}
  \label{eq: randomized numerical witness A}
  c_0 = 0.5,
  \quad
  c_Z = -1.0,
  \quad
  c_{ZZ} = -0.8,
  \quad
  \varepsilon = 0.05.
\end{equation}
For the preconditioner, we restrict to the diagonal Pauli ansatz
\begin{equation}
  \label{eq: randomized numerical ansatz}
  P_n = \sum_{s\in\{I,Z\}^{\otimes n}} \beta^{(n)}_s\, s,
\end{equation}
which keeps $P_n$ Hermitian and aligned with the dominant $I/Z$ structure of $A_n$.
In the computational basis, this $\{I,Z\}^{\otimes n}$ ansatz corresponds exactly to a diagonal (Jacobi-type) preconditioner, since every $n$-qubit operator diagonal in the Pauli $Z$ basis is a linear combination of $\{I,Z\}^{\otimes n}$ words.
This ansatz has \(2^n\) Pauli terms, so the experiment should be read as a
	proof-of-principle numerical witness for the improvement mechanism rather than
as a scalable construction.
Because a diagonal \(I/Z\) preconditioner need not commute with the \(XX+YY\)
part of \(A_n\), the product \(P_nA_n\) is generally non-Hermitian.  We therefore
interpret the randomized QLS solver through the Hermitian embedding
described above, and compute the table using the singular values of \(P_nA_n\)
and the regrouped Pauli coefficient weight \(w(P_nA_n)\).
For each qubit number $n$, we construct a candidate preconditioner by the
least-squares approximate inverse fit described in
Appendix~\ref{app: computation preconditioner coefficients}, and then certify its
effect by directly evaluating the proxy in~\eqref{eq: rqls-depth-proxy}.
We report the scale-invariant square root proxy ratio
\begin{equation}
  \label{eq: randomized numerical ratio}
  R_n
  :=
  \sqrt{
    \frac{D_{\mathrm{RQLS}}(P_n A_n)}{D_{\mathrm{RQLS}}(A_n)}
  }
  =
  \frac{w(P_n A_n)/\sigma_{\min}(P_n A_n)}
       {w(A_n)/\sigma_{\min}(A_n)}.
\end{equation}
Thus \(R_n^2\) is the preconditioned-to-unpreconditioned per-sample depth proxy
ratio.
The last two columns of Table~\ref{tab: randomized numerical witness} use the
same \(P_n\) values to evaluate the symmetric direct block-encoding comparison for
\(Q_n:=P_nA_nP_n\).  We set the direct block-encoding slack to
\(\delta_{\mathrm{be}}=0.1\) and use the norm-aware normalizations
\[
  \alpha_{A,n}^{\mathrm{dir}}
  =
  \frac{\|A_n\|}{1-\delta_{\mathrm{be}}},
  \quad
  \alpha_{P,n}
  =
  \frac{\|P_n\|}{1-\delta_{\mathrm{be}}},
  \quad
  \alpha_{Q,n}
  =
  \frac{\|Q_n\|}{1-\delta_{\mathrm{be}}}.
\]
The reported ratios are computed from
\[
\begin{aligned}
  \kappa_A^{\mathrm{dir}}
  &=
  \frac{\alpha_{A,n}^{\mathrm{dir}}}{\sigma_{\min}(A_n)},\\
  \kappa_Q^{\mathrm{dir}}
  &=
  \frac{\alpha_{Q,n}}{\sigma_{\min}(Q_n)},\\
  \kappa_{\mathrm{sep}}^{\mathrm{sym}}
  &=
  \frac{\alpha_{P,n}^2\alpha_{A,n}^{\mathrm{dir}}}{\sigma_{\min}(Q_n)}.
\end{aligned}
\]
These two columns are norm-aware effective QLS condition parameter diagnostics only.  They do
not include the cost of estimating the spectral norms, forming and storing the
regrouped Pauli list of \(Q_n\), or implementing one direct block-encoding query.

\begin{table*}[t]
  \centering
  \small
  \setlength{\tabcolsep}{4pt}
  \caption{Qubit sweep for the fixed \(A_n\) coefficient family, showing the randomized QLS per-sample depth proxy ratio \(R_n^2=D_{\mathrm{RQLS}}(P_nA_n)/D_{\mathrm{RQLS}}(A_n)\) and symmetric direct block-encoding effective QLS condition parameter diagnostics.}
  \label{tab: randomized numerical witness}
  \resizebox{0.8\textwidth}{!}{
  \begin{tabular}{c c c c c c c c c}
    \hline
    $n$ & $\sigma_{\min}(A_n)$ & $\sigma_{\min}(P_n A_n)$ & $w(A_n)$ & $w(P_n A_n)$ & $R_n$ & $R_n^2$ & $\kappa_Q^{\mathrm{dir}} / \kappa_A^{\mathrm{dir}}$ & $\kappa_Q^{\mathrm{dir}} / \kappa_{\mathrm{sep}}^{\mathrm{sym}}$ \\
    \hline
    $2$ & $1.200000$ & $0.229412$ & $3.4$ & $0.269118$ & $0.414027$ & $0.171418$ & $0.982422$ & $0.493043$ \\
    $3$ & $0.500000$ & $0.229811$ & $5.3$ & $0.351295$ & $0.144210$ & $0.020797$ & $0.935310$ & $0.099930$ \\
    $4$ & $0.303122$ & $0.196604$ & $7.2$ & $0.450454$ & $0.096459$ & $0.009304$ & $0.819487$ & $0.041240$ \\
    $5$ & $0.103131$ & $0.113104$ & $9.1$ & $0.739749$ & $0.074123$ & $0.005494$ & $0.258770$ & $0.010558$ \\
    $6$ & $0.096869$ & $0.105125$ & $11.0$ & $0.984098$ & $0.082437$ & $0.006796$ & $0.243173$ & $0.008562$ \\
    $7$ & $0.106256$ & $0.079689$ & $12.9$ & $1.454967$ & $0.150391$ & $0.022617$ & $0.287213$ & $0.017438$ \\
    \hline
  \end{tabular}
  }
\end{table*}

Table~\ref{tab: randomized numerical witness} shows that, for the fixed
diagonal-term-dominated \(A_n\) family, the separately computed preconditioners satisfy
\(R_n<1\) for all tested sizes \(n=2,\dots,7\).
Hence the randomized QLS solver improvement criterion in Remark~\ref{remark: rqls preconditioning benefit} is satisfied throughout this qubit range.
The improvement is substantial: the ratio drops to $R_5 \approx 0.074$ and remains below $0.16$ for $n=3,\dots,7$, corresponding to per-sample depth proxy reductions on the order of $10^{-2}$. The
reported values satisfy
\[
  \frac{w(P_nA_n)}{\sigma_{\min}(P_nA_n)}
  <
  \frac{w(A_n)}{\sigma_{\min}(A_n)}
\]
for every tested instance, which is exactly the condition \(R_n<1\).  Thus the
examples illustrate the regime predicted by
Theorem~\ref{thm: randomized qls preconditioning}: after regrouping the Pauli
expansion, the preconditioned operator has a smaller Pauli-weight-to-stability
ratio, and hence a smaller randomized QLS per-sample depth proxy.

The last two columns report the symmetric direct block-encoding comparison for
\(Q_n:=P_nA_nP_n\) from Section~\ref{sec: pauli structured data}.
For all tested sizes, the norm-aware diagnostic
\({\kappa_Q^{\mathrm{dir}}} / {\kappa_A^{\mathrm{dir}}}\) is below one, so the
classically formed symmetric product has a smaller effective QLS condition parameter
than the original matrix under the same norm-aware convention.
\section{Conclusion}
\label{sec: conclusion}

We have investigated quantum preconditioning under a Pauli-structured ansatz and obtained two complementary conclusions.
For standard QLS solvers based on block-encodings, the limitation result Theorem~\ref{thm:no-go-preconditioning} shows that separate block-encoding composition for $A$ and $P$ cannot reduce the contribution to the query complexity controlled by $\sigma_{\min}(A)$.
This clarifies a fundamental limitation of quantum preconditioning by oracle composition and explains why any genuine improvement for standard QLS solvers must arise from more structured access models, such as direct block-encodings of classically formed products.

Nevertheless, Pauli structure remains useful.
Under suitable commuting assumptions, Hamiltonian simulation combined with the matrix-logarithm construction yields explicit block-encodings of $P$, $PA$, and $PAP^\dagger$.
More significantly, for the randomized QLS solver in the randomized Pauli access model, the relevant complexity parameters are $w(PA)$ and $\sigma_{\min}(PA)$ rather than a block-encoding normalization factor; in that setting, Pauli-structured preconditioners can genuinely improve the dominant per-sample depth proxy.
Our numerical examples show that this per-sample advantage can be substantial even for a simple diagonal-term-dominated family, while remaining proof-of-principle numerical witnesses rather than scalable constructions; the full sample complexity also depends on the right-hand side normalization and on the chosen scalar functional.

Several directions remain open.
It would be valuable to characterize the regimes in which direct block-encodings of classically formed products can surpass the lower-bound mechanism for separate block-encoding composition identified here. Further open problems include the design of systematic optimization procedures for Pauli-structured preconditioners with small coefficient weight, the identification of polynomial-size local or translation-invariant Pauli ansatz families that retain the observed cancellations, and the connection between the Pauli coefficient weight proxy analyzed here and hardware-level gate cost models for concrete quantum architectures.

\begin{acknowledgments}
DA acknowledges funding from Quantum Science and Technology - National Science and Technology Major Project via Project 2024ZD0301900, and the support by The Fundamental Research Funds for the Central Universities, Peking University.
ZL acknowledges funding from the National Natural Science Foundation of China under Grant No. 12501419.
\end{acknowledgments}

\appendix

\section{Implementing the preconditioner via LCU}
\label{app: LCU implementation}

We now describe how the preconditioner~\eqref{eq: preconditioner}
admits a natural block-encoding via the linear combination of unitaries (LCU) framework~\cite{gilyen2019quantum,low2019hamiltonian,zhang2024circuit}, and how this leads to a complexity dependence on the Pauli coefficient weight $w(P)$.
In this section, we keep the standing assumption that $P$ is Hermitian, so the Pauli coefficients satisfy $\beta_m\in\mathbb{R}$.

The standard LCU framework is stated for nonnegative coefficients.
To handle general real coefficients, we absorb their signs into the Pauli words.
Write
\begin{equation}
  \beta_m = s_m \omega_m , \quad s_m \in \{\pm 1\}, \quad \omega_m = |\beta_m| \ge 0 .
\end{equation}
Then
$
 w(P) = \sum_{m=1}^{M} \omega_m
$.
Provided that $\beta \neq 0$, we introduce the probability weights
$
  \pi_m := \frac{\omega_m}{w(P)} , m = 1,\dots,M ,
$
and the sign-absorbed Pauli words
\begin{equation}
  \widetilde{P}_m := s_m P_m , \quad m = 1,\dots,M .
\end{equation}
Each $\widetilde{P}_m$ is again unitary and Hermitian, and the normalized operator is
$
  \frac{1}{w(P)} P = \sum_{m=1}^{M} \pi_m \widetilde{P}_m .
$

To construct a block-encoding of \(P\), we introduce an ancilla register \(\mathcal{H}_{\mathrm{anc}}\) of dimension at least \(M\), with computational basis \(\{ |m\rangle \}_{m=1}^{M}\).
The LCU construction uses two unitaries on the ancilla and system registers:
\begin{itemize}
  \item A state-preparation unitary \(U_{\mathrm{prep}}\) acting only on the ancilla such that
  \begin{equation}
    U_{\mathrm{prep}} |0\rangle
    = \sum_{m=1}^{M} \sqrt{\pi_m} \, |m\rangle
    = \sum_{m=1}^{M} \sqrt{\frac{|\beta_m|}{w(P)}} \, |m\rangle ,
  \end{equation}
  \item A SELECT oracle acting on ancilla and system as
  \begin{equation}
    \begin{aligned}
    \mathrm{SELECT}(P)
    &:= \sum_{m=1}^{M} |m\rangle \langle m| \otimes \widetilde{P}_m \\
    &= \sum_{m=1}^{M} |m\rangle \langle m| \otimes (s_m P_m) .
    \end{aligned}
  \end{equation}
\end{itemize}
Then the unitary $U_P$ is given by
  \begin{equation}
    \begin{aligned}
    U_P
    &:=
    (U_{\mathrm{prep}}^{\dagger} \otimes I)\,
    \mathrm{SELECT}(P)\,
    (U_{\mathrm{prep}} \otimes I).
    \end{aligned}
  \end{equation}
  This unitary is an \((w(P),1,0)\) block-encoding of the preconditioner \(P\).

Assume that we are given an \((\alpha_A,a_A,0)\) block-encoding \(U_A\) of the system matrix \(A\).
Up to trivial padding with identities on ancilla registers, the product unitary
$
 U_{PA} := U_P U_A
$
is then an \((w(P) \alpha_A,  a_A + 1, 0)\) block-encoding of the preconditioned operator \(P A\).

\section{Implementing the preconditioner via Hamiltonian simulation and matrix logarithm}
\label{app: hamiltonian simulation}

To interpret $P$ as a Hamiltonian, we assume in this subsection that $P$ is Hermitian.
Since each Pauli word $P_m$ is Hermitian, Hermiticity of $P$ requires $\beta_m \in \mathbb{R}$.
In this subsection, we implement the sum $P=\sum_{m=1}^M \beta_m P_m$ by first implementing $e^{itP}$ and then applying a matrix-logarithm construction.

\subsection{\texorpdfstring{Simulation of $U(t)=e^{itP}$ and its inverse}{Simulation of U(t)=exp(itP) and its inverse}}
Without loss of generality, we assume that the \(M\) listed coefficients are nonzero.
For a real time parameter \(t \in \mathbb{R}\), define the unitary
\begin{equation}
  U(t) := e^{itP} = \exp\!\left( i t \sum_{m=1}^{M} \beta_m P_m \right) .
\end{equation}
Set $H_m := i\beta_m P_m$ so that each $H_m$ is anti-Hermitian.
Define the $(p+1)$-fold commutator scaling
\begin{equation}
\label{eq:alpha-comm-scaling}
\tilde{\alpha}_{\mathrm{comm}}
:=\sum_{m_1,\ldots,m_{p+1}=1}^{M}
\left\|\left[H_{m_{p+1}},\left[\cdots\left[H_{m_2},H_{m_1}\right]\cdots\right]\right]\right\| .
\end{equation}
According to Corollary~7 of~\cite{childs2021theory}, for a product formula of order \(p\) of the form
$
\mathcal{S}(t) = \prod_{\upsilon=1}^{\Upsilon} \prod_{\gamma=1}^{M} \exp\!\bigl(t\, a(\upsilon,\gamma)\, H_{\pi_{\upsilon}(\gamma)}\bigr),
$
if we set
\begin{equation}
  r = \mathcal{O}\!\left(\frac{\tilde{\alpha}_{\mathrm{comm}}^{\frac{1}{p}}\, |t|^{1+\frac{1}{p}}}{\varepsilon_{\mathrm{exp}}^{\frac{1}{p}}}\right),
\end{equation}
then the simulation error satisfies
\begin{equation}
  \bigl\| \mathcal{S}^r\left(\frac{t}{r}\right) - U(t) \bigr\| \le \mathcal{O}(\varepsilon_{\mathrm{exp}}).
\end{equation}

For the inverse $U(t)^{\dagger}=U(-t)$, the same construction applies with $t \mapsto -t$ (equivalently $\beta_{\gamma} \mapsto -\beta_{\gamma}$).
We denote by $\widetilde{U}(t)$ the resulting approximate implementation of $U(t)$.

\subsection{Matrix-logarithm construction}
First, we state the following proposition, which improves the norm assumption of Corollary~71 of~\cite{gilyen2019quantum}.
\begin{proposition}
  \label{prop: log implementation}
  Let $\delta, \varepsilon \in\left(0, \frac{1}{2}\right]$.
Suppose that $U=e^{i H}$, where $H$ is a Hamiltonian of norm at most $\frac{\pi}{2}(1 - \delta)$.  Then, under our block-encoding normalization convention, we can implement a $\left(\frac{\pi}{2}, 2, \varepsilon\right)$-block-encoding of $H$ using $\mathcal{O}\left(\frac{1}{\delta^2}\log \left(\frac{1}{\varepsilon}\right)\right)$ calls to controlled-$U$ and its inverse, $\mathcal{O}\left(\frac{1}{\delta^2}\log \left(\frac{1}{\varepsilon}\right)\right)$ two-qubit gates, and a single ancilla qubit.
\end{proposition}
\begin{proof}
  Let $cU$ denote the controlled version of $U$ with a single control qubit.
  As in the original proof of Corollary~71 of~\cite{gilyen2019quantum}, one can express $\sin(H)$ as a matrix element of a
  constant-size circuit involving $cU$ and $cU^\dagger$:
  \[
  \sin(H)
  \;=\;
  -i\,(\langle +|\otimes I)\,cU^\dagger\,(ZX\otimes I)\,cU\,(|+\rangle\otimes I).
  \]
  Set $A:=\sin(H)$.  From $\|H\|\le \frac{\pi}{2}(1-\delta)$ and the fact that $\sin$ is increasing on $[0,\frac{\pi}{2}]$,
  we have
  $
  \|A\|
  \;\le\;
  \sin(\frac{\pi}{2}(1-\delta)).
  $
  Define the margin
  $
  \delta' \;:=\; 1-\sin(\frac{\pi}{2}(1-\delta))\in(0,1].
  $
  Then the spectrum of $A$ lies in $[-1+\delta',\,1-\delta']$.

  Now apply Lemma~70 with parameters $(\delta',\varepsilon)$.
  It gives an efficiently computable odd real polynomial $p\in\mathbb{R}[x]$ of degree
  $
  \deg(p) \;=\; \mathcal{O}\!\left(\frac{1}{\delta'}\log\!\left(\frac{1}{\varepsilon}\right)\right),
  $
  such that $\|p\|_{[-1,1]}\le 1$ and
  $
  \max_{x\in[-1+\delta',\,1-\delta']}
  \left|p(x)-\frac{2}{\pi}\arcsin(x)\right|
  \;\le\;
  \varepsilon.
  $
  Using quantum singular value transformation (QSVT; Corollary~18 in~\cite{gilyen2019quantum}) on the $(1,1,0)$-block-encoding of $A$
  with the polynomial $p$, we implement a $(1,2,\varepsilon)$-block-encoding of $p(A)$ using
  $\mathcal{O}(\deg(p))$ invocations of the block-encoding of $A$ and its inverse; hence using
  $\mathcal{O}(\deg(p))$ uses of controlled-$U$ and its inverse, and
  $\mathcal{O}(\deg(p))$ additional two-qubit gates.

  It remains to identify $p(A)$.  Since $\|H\|\le \frac{\pi}{2}(1-\delta) < \frac{\pi}{2}$, every eigenvalue $\lambda$ of $H$
  lies in $\left(-\frac{\pi}{2},\frac{\pi}{2}\right)$, and therefore $\arcsin(\sin(\lambda))=\lambda$. Hence,
  $
  \arcsin(A)
  \;=\;
  \arcsin(\sin(H))
  \;=\;
  H.
  $
  By functional calculus and the uniform approximation guarantee of $p$ on the spectrum of $A$,
  we obtain
  $
  \left\|p(A)-\frac{2}{\pi}\arcsin(A)\right\|
  \;\le\;
  \varepsilon.
  $
  Therefore,
  $
  \left\|p(\sin(H))-\frac{2}{\pi}H\right\|
  \;\le\;
  \varepsilon.
  $
  Thus the implemented unitary is a $(1,2,\varepsilon)$-block-encoding of $\left(\frac{2}{\pi}\right)H$, i.e., a
  $\left(\frac{\pi}{2},2,\varepsilon\right)$-block-encoding of $H$.

  Finally, to express the complexity in terms of \(\delta\), write
  \[
    \delta'
    =
    1-\sin\!\left(\frac{\pi}{2}(1-\delta)\right)
    =
    1-\cos\!\left(\frac{\pi\delta}{2}\right).
  \]
  For \(\delta\in(0,\frac{1}{2}]\), this margin satisfies
  \(\delta'=\Theta(\delta^2)\).  Therefore
  $
  \deg(p)
  \;=\;
  \mathcal{O}\!\left(\frac{1}{\delta^2}\log\!\left(\frac{1}{\varepsilon}\right)\right).
  $
  This completes the proof.
\end{proof}

Assume that we have an implementation of $U(t)$ and its inverse via the previous step.
Let $\delta \in \left(0, \frac{1}{2}\right]$.
We restrict the time parameter $t$ to the range
\begin{equation}
  \label{eq: time parameter range}
  0 < t \leq \frac{\pi (1-\delta)}{2\|P\|}.
\end{equation}

By Proposition~\ref{prop: log implementation}, for any $\varepsilon_{\mathrm{log}} \in\left(0, \frac{1}{2}\right]$, we can implement a $\left(\frac{\pi}{2}, 2, \varepsilon_{\mathrm{log}}\right)$-block-encoding of $(tP)$ using $\mathcal{O}\left(\frac{1}{\delta^2}\log \left(\frac{1}{\varepsilon_{\mathrm{log}}}\right)\right)$ calls to controlled-$U(t)$ and controlled-$U(t)^{\dagger}$, $\mathcal{O}\left(\frac{1}{\delta^2}\log \left(\frac{1}{\varepsilon_{\mathrm{log}}}\right)\right)$ two-qubit gates, and one ancilla qubit.

\paragraph{Cumulative error and total complexity.}
The construction in Corollary~71 of~\cite{gilyen2019quantum} has two stages.
First, one block-encodes $\sin(tP)$ using controlled-$U(t)$ and controlled-$U(t)^{\dagger}$.
Second, one applies a QSVT step with an \(\varepsilon_{\log}\)-approximating
polynomial for \(\frac{2}{\pi}\arcsin(x)\) on the spectral interval of
\(\sin(tP)\), which stays a distance \(\Theta(\delta^2)\) from \(\pm1\).
Let
\begin{equation}
  \begin{aligned}
  W_{\sin}
  &:= \sin(tP) \\
  &=
  -i\,(\langle +|\otimes I)\,cU(t)^{\dagger}(ZX\otimes I)\,cU(t)\,(|+\rangle\otimes I)
  \end{aligned}
\end{equation}
be the ideal block-encoding map for $\sin(tP)$, and let $\widetilde{W}_{\sin}$ be the same expression with $cU(t)$ and $cU(t)^{\dagger}$ replaced by $c\widetilde{U}(t)$ and $c\widetilde{U}(t)^{\dagger}$.
Since controlling does not change operator norm,
\begin{equation}
  \begin{aligned}
  \|c\widetilde{U}(t) - cU(t)\|
  &= \|\widetilde{U}(t) - U(t)\|
  \le \varepsilon_{\mathrm{exp}}, \\
  \|c\widetilde{U}(t)^{\dagger} - cU(t)^{\dagger}\|
  &\le \varepsilon_{\mathrm{exp}}.
  \end{aligned}
\end{equation}
A direct telescoping bound using the fact that all other factors in $W_{\sin}$ are unitary gives
\begin{equation}
  \begin{aligned}
  \|\widetilde{W}_{\sin} - W_{\sin}\|
  &\le
  \|c\widetilde{U}(t)^{\dagger} - cU(t)^{\dagger}\| \\
  &\quad + \|c\widetilde{U}(t) - cU(t)\|
  \le 2\varepsilon_{\mathrm{exp}}.
  \end{aligned}
\end{equation}
Thus the implemented block-encoding of $\sin(tP)$ has intrinsic error
$
  \varepsilon_{\sin} = \mathcal{O}(\varepsilon_{\mathrm{exp}}).
$

Let $q$ denote the degree (and query complexity) of this QSVT polynomial, so that
$
  q = \mathcal{O}\!\left(\frac{1}{\delta^2}\log\!\left(\tfrac{1}{\varepsilon_{\log}}\right)\right).
$
Robustness of QSVT implies that when the signal unitary (here, the block-encoding for $\sin(tP)$) has error $\varepsilon_{\sin}$, the induced error after a degree-$q$ transformation accumulates at most linearly in $q$.
Consequently the overall error of the resulting $\left(\tfrac{\pi}{2},2,\varepsilon_{\mathrm{tot}}\right)$-block-encoding of $tP$ satisfies
\begin{equation}
  \varepsilon_{\mathrm{tot}} \le \varepsilon_{\log} + \mathcal{O}(q\,\varepsilon_{\sin})
  \le \varepsilon_{\log} + \mathcal{O}(q\,\varepsilon_{\mathrm{exp}}).
\end{equation}
This block-encoding is equivalent to a $\left(\frac{\pi}{2 t}, 2, \varepsilon_{\mathrm{tot}}\right)$-block-encoding of $P$.

Each call to $\widetilde{U}(t)$ (or its inverse) uses $\mathcal{O}(r\,\Upsilon\,M)$ Pauli exponentials, and the controlled versions incur only constant-factor overhead.
Implementing one use of the $\sin$ block-encoding uses two controlled calls (one to $\widetilde{U}(t)$ and one to $\widetilde{U}(t)^{\dagger}$), so the overall Pauli exponential count remains
$
  \mathcal{O}\bigl(q\, r\, \Upsilon\,M\bigr)
$
up to constant factors, plus $\mathcal{O}(q)$ additional two-qubit gates and a single ancilla qubit as in Corollary~71 of~\cite{gilyen2019quantum}.
For the Suzuki formula $S_{2k}(t)$ (order $p=2k$), one has $\Upsilon = 2 \cdot 5^{k-1} = 2 \cdot 5^{\frac{p}{2} - 1}$.

Let $C_0$ denote the absolute constant in the QSVT robustness bound
$\varepsilon_{\mathrm{tot}} \le \varepsilon_{\log} + C_0\, q\, \varepsilon_{\sin}$
(see~\cite{gilyen2019quantum}).
Set
$
  \varepsilon_{\log} = \tfrac{1}{2}\varepsilon_{\mathrm{tot}},
  \qquad
  \varepsilon_{\mathrm{exp}} = \tfrac{\varepsilon_{\mathrm{tot}}}{4C_0 q}.
$
Since $\varepsilon_{\sin} = \mathcal{O}(\varepsilon_{\mathrm{exp}})$ from the telescoping bound above,
$C_0\,q\,\varepsilon_{\sin} \le \tfrac{1}{2}\varepsilon_{\mathrm{tot}}$,
and therefore
$\varepsilon_{\mathrm{tot}} \le \varepsilon_{\log} + C_0\,q\,\varepsilon_{\sin} \le \tfrac{1}{2}\varepsilon_{\mathrm{tot}} + \tfrac{1}{2}\varepsilon_{\mathrm{tot}} = \varepsilon_{\mathrm{tot}}$,
confirming the allocation is consistent.
Substituting this choice into the Step~1 bound on $r$ gives
\begin{equation}
  r = \mathcal{O}\!\left(\tilde{\alpha}_{\mathrm{comm}}^{\frac{1}{p}}\, t^{1+\frac{1}{p}}\,\left(\frac{q}{\varepsilon_{\mathrm{tot}}}\right)^{\frac{1}{p}}\right),
\end{equation}
and therefore the total Pauli exponential complexity becomes
\begin{equation}
  \label{eq: complexity of P}
  \begin{aligned}
    \mathrm{T}_{P}
    &= \mathcal{O}\!\left(
      \Upsilon\,M\,\tilde{\alpha}_{\mathrm{comm}}^{\frac{1}{p}}\,
      t^{1+\frac{1}{p}}\,
      \frac{q^{1+\frac{1}{p}}}{\varepsilon_{\mathrm{tot}}^{\frac{1}{p}}}
    \right) \\
	    &= \mathcal{O}\!\left(
	      \Upsilon\,M\,\tilde{\alpha}_{\mathrm{comm}}^{\frac{1}{p}}\,
	      t^{1+\frac{1}{p}}\,
	      \frac{\bigl(\log(\frac{1}{\varepsilon_{\mathrm{tot}}})\bigr)^{1+\frac{1}{p}}}
             {\delta^{2(1+\frac{1}{p})}\varepsilon_{\mathrm{tot}}^{\frac{1}{p}}}
	    \right) \\
	    &= \widetilde{O}\left(
	      \Upsilon\,M\,\tilde{\alpha}_{\mathrm{comm}}^{\frac{1}{p}}\,
	      t^{1+\frac{1}{p}}\,
        \delta^{-2(1+\frac{1}{p})}\,
	      \frac{1}{\varepsilon_{\mathrm{tot}}^{\frac{1}{p}}}
	    \right)
  \end{aligned}
\end{equation}
By taking the maximum allowed value of $t$ as $t = \frac{\pi (1-\delta)}{2\|P\|}$, we have
\begin{equation}
	  \mathrm{T}_P = \widetilde{O}\left(
	    \Upsilon\,M\,\tilde{\alpha}_{\mathrm{comm}}^{\frac{1}{p}}\, \left(\frac{\pi (1-\delta)}{2\|P\|}\right)^{1+\frac{1}{p}}\,
      \delta^{-2(1+\frac{1}{p})}\,
	    \frac{1}{\varepsilon_{\mathrm{tot}}^{\frac{1}{p}}}
	  \right)
\end{equation}

\section{Computation of the preconditioner coefficients}
\label{app: computation preconditioner coefficients}

We now describe how the Pauli coefficients used in the numerical experiment are
computed.  Fix an \(n\)-qubit instance \(A_n\) and a finite Pauli ansatz
\begin{equation}
  \label{eq: coefficient computation ansatz}
  P(\beta)=\sum_{j=1}^{M_n}\beta_j S_j,
  \qquad
  \beta_j\in\mathbb{R},
\end{equation}
where, in the experiment of Table~\ref{tab: randomized numerical witness}, the
basis is the diagonal Pauli family
\begin{equation}
  \label{eq: coefficient computation diagonal basis}
  \mathcal{B}_n=\{S_j\}_{j=1}^{M_n}=\{I,Z\}^{\otimes n}.
\end{equation}
The coefficients are chosen by fitting \(P(\beta)A_n\) to the identity.  After
vectorization, define
\begin{equation}
  \label{eq: coefficient computation feature matrix}
  F_n =
  \begin{bmatrix}
    \operatorname{vec}(S_1A_n) & \cdots & \operatorname{vec}(S_{M_n}A_n)
  \end{bmatrix},
  \qquad
  e_n=\operatorname{vec}(I).
\end{equation}
For real coefficients, if a general Pauli ansatz produces complex entries, one
may replace \(F_n\beta\approx e_n\) by its realification
\[
  \begin{bmatrix}\operatorname{Re}F_n\\ \operatorname{Im}F_n\end{bmatrix}\beta
  \approx
  \begin{bmatrix}e_n\\0\end{bmatrix}.
\]
In the diagonal ansatz used for Table~\ref{tab: randomized numerical witness},
the matrices \(S_jA_n\) are real, so this realification reduces to the real
least-squares system.

Candidate coefficient vectors are obtained from the approximate inverse
problem
\begin{equation}
  \label{eq: coefficient computation lasso}
  \widehat{\beta}_{\lambda}
  \in
  \argmin_{\beta\in\mathbb{R}^{M_n}}
  \frac{1}{2}\|F_n\beta-e_n\|_2^2+\lambda\|\beta\|_1.
\end{equation}
This is a LASSO objective~\cite{tibshirani1996regression}.  The unregularized case \(\lambda=0\) is the ordinary least-squares fit.  For
\(\lambda>0\), the \(\ell_1\) term promotes a Pauli expansion with smaller
coefficient weight and can be solved by the standard accelerated proximal
gradient method, also known as FISTA~\cite{beck2009fast}; see Algorithm~\ref{alg: coefficient computation fista}.  If \(L_F=\|F_n\|_2^2\), the smooth part
\(\frac12\|F_n\beta-e_n\|_2^2\) has \(L_F\) Lipschitz gradient
\[
  \nabla f(\beta)=F_n^{T}(F_n\beta-e_n).
\]
The proximal step is coordinatewise soft thresholding:
\begin{equation}
  \label{eq: coefficient computation soft threshold}
  \mathcal{S}_{\tau}(x)_j
  =
  \operatorname{sign}(x_j)\max\{|x_j|-\tau,0\}.
\end{equation}

\begin{figure*}[t]
\begingroup
\renewcommand{\figurename}{Algorithm}
\caption{Coefficient computation for Pauli-structured preconditioners}
\label{alg: coefficient computation fista}
\centering
\setlength{\fboxsep}{6pt}
\setlength{\fboxrule}{0.5pt}
\fcolorbox{black!55}{white}{
\begin{minipage}{0.94\textwidth}
\small
\vspace{0.3ex}
\begin{algorithmic}[1]
\REQUIRE Pauli basis \(\mathcal{B}_n=\{S_j\}_{j=1}^{M_n}\), matrix \(A_n\), regularization \(\lambda\ge0\), optional positive display scale \(\gamma\), iteration budget \(K\)
\ENSURE A Pauli-structured preconditioner \(P_n\)
\STATE Form \(F_n=[\operatorname{vec}(S_1A_n)\ \cdots\ \operatorname{vec}(S_{M_n}A_n)]\) and \(e_n=\operatorname{vec}(I)\).
\IF{\(\lambda=0\)}
  \STATE Compute \(\widehat{\beta}\in\argmin_\beta \|F_n\beta-e_n\|_2^2\) by a least-squares solve.
\ELSE
  \STATE Set \(\beta^{0}=z^{0}=0\), \(\theta_0=1\), and \(L_F=\|F_n\|_2^2\).
  \FOR{\(k=0,1,\ldots,K-1\)}
    \STATE \(g^k=F_n^T(F_nz^k-e_n)\).
    \STATE \(\beta^{k+1}=\mathcal{S}_{\frac{\lambda}{L_F}}(z^k-\frac{g^k}{L_F})\).
    \STATE \(\theta_{k+1}=\frac{1+\sqrt{1+4\theta_k^2}}{2}\).
    \STATE \(z^{k+1}=\beta^{k+1}+\frac{\theta_k-1}{\theta_{k+1}}(\beta^{k+1}-\beta^k)\).
  \ENDFOR
  \STATE Set \(\widehat{\beta}=\beta^K\).
\ENDIF
\STATE Form \(\widehat P=\sum_j\widehat{\beta}_jS_j\), regroup \(\widehat PA_n\), and compute the depth ratio \(R(\widehat P)\).
\STATE Optionally rescale for reporting by setting \(P_n=\gamma\widehat P\) with \(\gamma>0\); the reported ratios are unchanged by this positive rescaling.
\STATE Return \(P_n\) together with the directly evaluated diagnostic \(R(P_n)\).
\end{algorithmic}
\vspace{0.3ex}
\end{minipage}
}
\endgroup
\end{figure*}

The diagnostic evaluation step is important.  The optimization objective
\eqref{eq: coefficient computation lasso} is used only to construct candidate
preconditioners; the improvement reported in Table~\ref{tab: randomized numerical witness}
is certified afterwards by direct evaluation of \(w(P_nA_n)\),
\(\sigma_{\min}(P_nA_n)\), and
\[
  R_n=
  \sqrt{\frac{D_{\mathrm{RQLS}}(P_nA_n)}{D_{\mathrm{RQLS}}(A_n)}}
  =
  \frac{w(P_nA_n)/\sigma_{\min}(P_nA_n)}
       {w(A_n)/\sigma_{\min}(A_n)}.
\]
For the table, the listed \(P_n\)'s come from the least-squares branch
\(\lambda=0\), followed by a fixed positive rescaling used only to set a
convenient coefficient scale.  This rescaling is not a screening step: the
ratios \(R_n\), \(\frac{\kappa_Q^{\mathrm{dir}}}{\kappa_A^{\mathrm{dir}}}\), and
\(\frac{\kappa_Q^{\mathrm{dir}}}{\kappa_{\mathrm{sep}}^{\mathrm{sym}}}\) are invariant
under \(P_n\mapsto \gamma P_n\) for \(\gamma>0\).  The same coefficients are
then reused, without reoptimization, to form \(Q_n=P_nA_nP_n\) for the
symmetric direct block-encoding comparison.

\bibliography{ref}

@article{harrow2009quantum,
  title         = {Quantum Algorithm for Linear Systems of Equations},
  author        = {Harrow, Aram W. and Hassidim, Avinatan and Lloyd, Seth},
  journal       = {Physical Review Letters},
  volume        = {103},
  number        = {15},
  pages         = {150502},
  year          = {2009},
  doi           = {10.1103/PhysRevLett.103.150502}
}

@article{childs2017quantum,
  title         = {Quantum Algorithm for Systems of Linear Equations with Exponentially Improved Dependence on Precision},
  author        = {Childs, Andrew M. and Kothari, Robin and Somma, Rolando D.},
  journal       = {SIAM Journal on Computing},
  volume        = {46},
  number        = {6},
  pages         = {1920--1950},
  year          = {2017},
  doi           = {10.1137/16M1087072}
}

@inproceedings{gilyen2019quantum,
  title         = {Quantum Singular Value Transformation and Beyond: Exponential Improvements for Quantum Matrix Arithmetics},
  author        = {Gily{\'e}n, Andr{\'a}s and Su, Yuan and Low, Guang Hao and Wiebe, Nathan},
  booktitle     = {Proceedings of the 51st Annual ACM SIGACT Symposium on Theory of Computing},
  pages         = {193--204},
  year          = {2019},
  doi           = {10.1145/3313276.3316366}
}

@book{saad2003iterative,
  title         = {Iterative Methods for Sparse Linear Systems},
  author        = {Saad, Yousef},
  edition       = {2},
  publisher     = {SIAM},
  address       = {Philadelphia},
  year          = {2003},
  doi           = {10.1137/1.9780898718003}
}

@article{benzi2002preconditioning,
  title         = {Preconditioning Techniques for Large Linear Systems: A Survey},
  author        = {Benzi, Michele},
  journal       = {Journal of Computational Physics},
  volume        = {182},
  number        = {2},
  pages         = {418--477},
  year          = {2002},
  doi           = {10.1006/jcph.2002.7176}
}

@article{grote1997parallel,
  title         = {Parallel Preconditioning with Sparse Approximate Inverses},
  author        = {Grote, Marcus J. and Huckle, Thomas},
  journal       = {SIAM Journal on Scientific Computing},
  volume        = {18},
  number        = {3},
  pages         = {838--853},
  year          = {1997},
  doi           = {10.1137/S1064827594276552}
}

@article{clader2013preconditioned,
  title         = {Preconditioned Quantum Linear System Algorithm},
  author        = {Clader, B. David and Jacobs, Bryan C. and Sprouse, Chad R.},
  journal       = {Physical Review Letters},
  volume        = {110},
  number        = {25},
  pages         = {250504},
  year          = {2013},
  doi           = {10.1103/PhysRevLett.110.250504}
}

@article{shao2018quantum,
  title         = {Quantum Circulant Preconditioner for a Linear System of Equations},
  author        = {Shao, Changpeng and Xiang, Hua},
  journal       = {Physical Review A},
  volume        = {98},
  number        = {6},
  pages         = {062321},
  year          = {2018},
  doi           = {10.1103/PhysRevA.98.062321}
}

@article{wan2018asymptotic,
  title         = {Asymptotic Quantum Algorithm for the {Toeplitz} Systems},
  author        = {Wan, Lin-Chun and Yu, Chao-Hua and Pan, Shi-Jie and Gao, Fei and Wen, Qiao-Yan and Qin, Su-Juan},
  journal       = {Physical Review A},
  volume        = {97},
  number        = {6},
  pages         = {062322},
  year          = {2018},
  doi           = {10.1103/PhysRevA.97.062322}
}

@article{tong2021fast,
  title         = {Fast Inversion, Preconditioned Quantum Linear System Solvers, and Fast Evaluation of Matrix Functions},
  author        = {Tong, Yimin and An, Dong and Wiebe, Nathan and Lin, Lin},
  journal       = {Physical Review A},
  volume        = {104},
  number        = {3},
  pages         = {032422},
  year          = {2021},
  doi           = {10.1103/PhysRevA.104.032422}
}

@article{golden2022quantum,
  title         = {Quantum Computing and Preconditioners for Hydrological Linear Systems},
  author        = {Golden, Jacob K. and O'Malley, Daniel and Viswanathan, Hari S.},
  journal       = {Scientific Reports},
  volume        = {12},
  number        = {1},
  pages         = {22285},
  year          = {2022},
  doi           = {10.1038/s41598-022-25727-9}
}

@misc{hosaka2023preconditioning,
  title         = {Preconditioning for a Variational Quantum Linear Solver},
  author        = {Hosaka, Aruto and Yanagisawa, Koichi and Koshikawa, Shota and Kudo, Isamu and Alifu, Xiafukaiti and Yoshida, Tsuyoshi},
  year          = {2023},
  eprint        = {2312.15657},
  archivePrefix = {arXiv},
  doi           = {10.48550/arXiv.2312.15657}
}

@article{wu2024preconditioned,
  title         = {A Preconditioned Inexact Infeasible Quantum Interior Point Method for Linear Optimization},
  author        = {Wu, Zeguan and Yang, Xiu and Terlaky, Tam{\'a}s},
  journal       = {Computational Optimization and Applications},
  volume        = {93},
  number        = {3},
  pages         = {1225--1257},
  year          = {2026},
  doi           = {10.1007/s10589-025-00750-4}
}

@book{nielsen2010quantum,
  title         = {Quantum Computation and Quantum Information},
  author        = {Nielsen, Michael A. and Chuang, Isaac L.},
  edition       = {10th anniversary},
  publisher     = {Cambridge University Press},
  address       = {Cambridge},
  year          = {2010},
  doi           = {10.1017/CBO9780511976667}
}

@article{tibshirani1996regression,
  title         = {Regression Shrinkage and Selection via the {LASSO}},
  author        = {Tibshirani, Robert},
  journal       = {Journal of the Royal Statistical Society. Series B},
  volume        = {58},
  number        = {1},
  pages         = {267--288},
  year          = {1996},
  doi           = {10.1111/j.2517-6161.1996.tb02080.x}
}

@article{beck2009fast,
  title         = {A Fast Iterative Shrinkage-Thresholding Algorithm for Linear Inverse Problems},
  author        = {Beck, Amir and Teboulle, Marc},
  journal       = {SIAM Journal on Imaging Sciences},
  volume        = {2},
  number        = {1},
  pages         = {183--202},
  year          = {2009},
  doi           = {10.1137/080716542}
}

@article{phillips2021quantum,
  title         = {{QW-HHL}: A Quantum Computer Amenable General Matrix Equation Solver},
  author        = {Phillips, Christopher D. and Akbari, Helia and Okhmatovski, Vladimir I.},
  journal       = {{IEEE} Journal on Multiscale and Multiphysics Computational Techniques},
  volume        = {11},
  pages         = {156--178},
  year          = {2026},
  doi           = {10.1109/JMMCT.2026.3678506}
}

@misc{jin2025schrodingerization,
  title         = {Quantum Preconditioning Method for Linear Systems Problems via {Schr{\"o}dingerization}},
  author        = {Jin, Shi and Liu, Nana and Ma, Chuwen and Yu, Yue},
  year          = {2025},
  eprint        = {2505.06866},
  archivePrefix = {arXiv},
  doi           = {10.48550/arXiv.2505.06866}
}

@article{shayegan2025quasi,
  title         = {Quantum Linear System Algorithm for Solving an Ill-Posed Quasi-Linear Elliptic Problem by Preconditioning Operator},
  author        = {Salehi Shayegan, Amir Hossein and Dejam, Laya},
  journal       = {The European Physical Journal Plus},
  volume        = {140},
  number        = {5},
  pages         = {467},
  year          = {2025},
  doi           = {10.1140/epjp/s13360-025-06383-0}
}

@article{suresh2023sparse,
  title         = {Computing a Sparse Approximate Inverse on Quantum Annealing Machines},
  author        = {Suresh, Sanjay and Suresh, Krishnan},
  journal       = {Journal of Computing and Information Science in Engineering},
  volume        = {26},
  number        = {1},
  pages         = {011008},
  year          = {2026},
  doi           = {10.1115/1.4070578}
}

@article{raisuddin2024review,
  title         = {A Review of Quantum Scientific Computing Algorithms Relevant to Computational Mechanics},
  author        = {Raisuddin, Osama Muhammad and De, Suvranu},
  journal       = {Archives of Computational Methods in Engineering},
  volume        = {33},
  number        = {1},
  pages         = {745--787},
  year          = {2026},
  doi           = {10.1007/s11831-025-10321-9}
}

@article{morales2024quantum,
  title         = {Quantum Linear System Solvers: A Survey of Algorithms and Applications},
  author        = {Morales, Mauro E. S. and Pira, Lirand{\"e} and Schleich, Philipp and Koor, Kelvin and Costa, Pedro and An, Dong and Aspuru-Guzik, Al{\'a}n and Lin, Lin and Rebentrost, Patrick and Berry, Dominic W.},
  journal       = {Reviews of Modern Physics},
  year          = {2026},
  doi           = {10.1103/x6gh-d8gh}
}

@article{childs2021theory,
  title         = {Theory of {Trotter} Error with Commutator Scaling},
  author        = {Childs, Andrew M. and Su, Yuan and Tran, Minh C. and Wiebe, Nathan and Zhu, Shuchen},
  journal       = {Physical Review X},
  volume        = {11},
  number        = {1},
  pages         = {011020},
  year          = {2021},
  doi           = {10.1103/PhysRevX.11.011020}
}

@article{wang2024qubit,
  title         = {Qubit-Efficient Randomized Quantum Algorithms for Linear Algebra},
  author        = {Wang, Samson and McArdle, Sam and Berta, Mario},
  journal       = {{PRX} Quantum},
  volume        = {5},
  number        = {2},
  pages         = {020324},
  year          = {2024},
  doi           = {10.1103/PRXQuantum.5.020324}
}

@article{lapworth2025preconditioned,
  author  = {Leigh Lapworth and Christoph S{\"u}nderhauf},
  title   = {Preconditioned block encodings for quantum linear systems},
  journal = {Quantum Science and Technology},
  volume  = {10},
  number  = {4},
  pages   = {045064},
  year    = {2025},
  doi     = {10.1088/2058-9565/ae0f4b}
}

@article{pechan2025block,
  title         = {Block Encoding of the Three-Dimensional Heterogeneous {Poisson} Equation with Application to Fracture Flow},
  author        = {Pechan, Austin and Golden, John and O'Malley, Daniel},
  journal       = {Physical Review Applied},
  volume        = {25},
  number        = {4},
  pages         = {044038},
  year          = {2026},
  doi           = {10.1103/mx5c-8vqk}
}

@article{childs2021high,
  title         = {High-Precision Quantum Algorithms for Partial Differential Equations},
  author        = {Childs, Andrew M. and Liu, Jin-Peng and Ostrander, Aaron},
  journal       = {Quantum},
  volume        = {5},
  pages         = {574},
  year          = {2021},
  doi           = {10.22331/q-2021-11-10-574}
}

@article{sato2024hamiltonian,
  title         = {Hamiltonian Simulation for Hyperbolic Partial Differential Equations by Scalable Quantum Circuits},
  author        = {Sato, Yuki and Kondo, Ruho and Hamamura, Ikko and Onodera, Tamiya and Yamamoto, Naoki},
  journal       = {Physical Review Research},
  volume        = {6},
  number        = {3},
  pages         = {033246},
  year          = {2024},
  doi           = {10.1103/PhysRevResearch.6.033246}
}

@misc{sturm2025efficient,
  title         = {Efficient and Explicit Block Encoding of Finite Difference Discretizations of the {Laplacian}},
  author        = {Sturm, Andreas and Schillo, Niclas},
  year          = {2025},
  eprint        = {2509.02429},
  archivePrefix = {arXiv},
  doi           = {10.48550/arXiv.2509.02429}
}

@article{lloyd1996universal,
  title         = {Universal Quantum Simulators},
  author        = {Lloyd, Seth},
  journal       = {Science},
  volume        = {273},
  number        = {5278},
  pages         = {1073--1078},
  year          = {1996},
  doi           = {10.1126/science.273.5278.1073}
}

@article{somma2002simulating,
  title         = {Simulating Physical Phenomena by Quantum Networks},
  author        = {Somma, Rolando and Ortiz, Gerardo and Gubernatis, James E. and Knill, Emanuel and Laflamme, Raymond},
  journal       = {Physical Review A},
  volume        = {65},
  number        = {4},
  pages         = {042323},
  year          = {2002},
  doi           = {10.1103/PhysRevA.65.042323}
}

@article{jordan1928uber,
  title         = {{\"U}ber das {Paulische} {{\"A}quivalenzverbot}},
  author        = {Jordan, Pascual and Wigner, Eugene},
  journal       = {Zeitschrift f{\"u}r Physik},
  volume        = {47},
  number        = {9},
  pages         = {631--651},
  year          = {1928},
  doi           = {10.1007/BF01331938}
}

@article{mcardle2020quantum,
  title         = {Quantum Computational Chemistry},
  author        = {McArdle, Sam and Endo, Suguru and Aspuru-Guzik, Al{\'a}n and Benjamin, Simon C. and Yuan, Xiao},
  journal       = {Reviews of Modern Physics},
  volume        = {92},
  number        = {1},
  pages         = {015003},
  year          = {2020},
  doi           = {10.1103/RevModPhys.92.015003}
}

@article{cao2019quantum,
  title         = {Quantum Chemistry in the Age of Quantum Computing},
  author        = {Cao, Yudong and Romero, Jonathan and Olson, Jonathan P. and Degroote, Matthias and Johnson, Peter D. and Kieferov{\'a}, M{\'a}ria and Kivlichan, Ian D. and Menke, Tim and Peropadre, Borja and Sawaya, Nicolas P. D. and Sim, Sukin and Veis, Libor and Aspuru-Guzik, Al{\'a}n},
  journal       = {Chemical Reviews},
  volume        = {119},
  number        = {19},
  pages         = {10856--10915},
  year          = {2019},
  doi           = {10.1021/acs.chemrev.8b00803}
}

@article{sunderhauf2024blockencoding,
  title   = {Block-Encoding Structured Matrices for Data Input in Quantum Computing},
  author  = {S{\"u}nderhauf, Christoph and Campbell, Earl and Camps, Joan},
  journal = {Quantum},
  volume  = {8},
  pages   = {1226},
  year    = {2024},
  doi     = {10.22331/q-2024-01-11-1226}
}

@article{camps2024explicit,
  title   = {Explicit Quantum Circuits for Block Encodings of Certain Sparse Matrices},
  author  = {Camps, Daan and Lin, Lin and Van Beeumen, Roel and Yang, Chao},
  journal = {SIAM Journal on Matrix Analysis and Applications},
  volume  = {45},
  number  = {1},
  pages   = {801--827},
  year    = {2024},
  doi     = {10.1137/22M1484298}
}

@article{zhang2024circuit,
  title   = {Circuit Complexity of Quantum Access Models for Encoding Classical Data},
  author  = {Zhang, Xiao-Ming and Yuan, Xiao},
  journal = {npj Quantum Information},
  volume  = {10},
  pages   = {42},
  year    = {2024},
  doi     = {10.1038/s41534-024-00835-8}
}

@article{low2019hamiltonian,
  title   = {Hamiltonian Simulation by Qubitization},
  author  = {Low, Guang Hao and Chuang, Isaac L.},
  journal = {Quantum},
  volume  = {3},
  pages   = {163},
  year    = {2019},
  doi     = {10.22331/q-2019-07-12-163}
}

@misc{schillo2026block,
  title         = {{TARE}: Block Encoding Linear Combinations of {Pauli} Strings Without Ancilla State Preparation},
  author        = {Schillo, Niclas and Sturm, Andreas and Quay, R{\"u}diger},
  year          = {2026},
  eprint        = {2601.05740},
  archivePrefix = {arXiv},
  doi           = {10.48550/arXiv.2601.05740}
}

@misc{hariprakash2025randomized,
  title         = {Are Randomized Quantum Linear Systems Solvers Practical?},
  author        = {Hariprakash, Siddharth and Van Beeumen, Roel and Klymko, Katherine and Camps, Daan},
  year          = {2025},
  eprint        = {2510.13766},
  archivePrefix = {arXiv},
  doi           = {10.48550/arXiv.2510.13766}
}

@article{berry2015simulating,
  title         = {Simulating {Hamiltonian} Dynamics with a Truncated {Taylor} Series},
  author        = {Berry, Dominic W. and Childs, Andrew M. and Cleve, Richard and Kothari, Robin and Somma, Rolando D.},
  journal       = {Physical Review Letters},
  volume        = {114},
  pages         = {090502},
  year          = {2015},
  doi           = {10.1103/PhysRevLett.114.090502}
}

@article{low2017optimal,
  title         = {Optimal {Hamiltonian} Simulation by Quantum Signal Processing},
  author        = {Low, Guang Hao and Chuang, Isaac L.},
  journal       = {Physical Review Letters},
  volume        = {118},
  pages         = {010501},
  year          = {2017},
  doi           = {10.1103/PhysRevLett.118.010501}
}

@article{lin2020optimal,
  title         = {Optimal polynomial based quantum eigenstate filtering with application to solving quantum linear systems},
  author        = {Lin, Lin and Tong, Yu},
  journal       = {Quantum},
  volume        = {4},
  pages         = {361},
  year          = {2020},
  doi           = {10.22331/q-2020-11-11-361}
}

@article{costa2022optimal,
  title         = {Optimal Scaling Quantum Linear-Systems Solver via Discrete Adiabatic Theorem},
  author        = {Costa, Pedro C. S. and An, Dong and Sanders, Yuval R. and Su, Yuan and Babbush, Ryan and Berry, Dominic W.},
  journal       = {{PRX} Quantum},
  volume        = {3},
  pages         = {040303},
  year          = {2022},
  doi           = {10.1103/PRXQuantum.3.040303}
}

@article{martyn2021grand,
  title         = {Grand Unification of Quantum Algorithms},
  author        = {Martyn, John M. and Rossi, Zane M. and Tan, Andrew K. and Chuang, Isaac L.},
  journal       = {{PRX} Quantum},
  volume        = {2},
  pages         = {040203},
  year          = {2021},
  doi           = {10.1103/PRXQuantum.2.040203}
}

@article{babbush2019quantum,
  title         = {Quantum Simulation of Chemistry with Sublinear Scaling in Basis Size},
  author        = {Babbush, Ryan and Berry, Dominic W. and McClean, Jarrod R. and Neven, Hartmut},
  journal       = {npj Quantum Information},
  volume        = {5},
  pages         = {92},
  year          = {2019},
  doi           = {10.1038/s41534-019-0199-y}
}

@article{chakraborty2024lcu,
  author        = {Chakraborty, Shantanav},
  title         = {Implementing any {Linear Combination of Unitaries} on Intermediate-term Quantum Computers},
  journal       = {Quantum},
  volume        = {8},
  pages         = {1496},
  year          = {2024},
  doi           = {10.22331/q-2024-10-10-1496},
  eprint        = {2302.13555},
  archivePrefix = {arXiv},
  primaryClass  = {quant-ph}
}

@misc{chakraborty2025qsvtwithout,
  author        = {Chakraborty, Shantanav and Hazra, Soumyabrata and Li, Tongyang and Shao, Changpeng and Wang, Xinzhao and Zhang, Yuxin},
  title         = {{Quantum Singular Value Transformation} without Block Encodings: Near-optimal Complexity with Minimal Ancilla},
  year          = {2025},
  eprint        = {2504.02385},
  archivePrefix = {arXiv},
  primaryClass  = {quant-ph}
}

@misc{wang2025randomizedqsvt,
  author        = {Wang, Xinzhao and Zhang, Yuxin and Hazra, Soumyabrata and Li, Tongyang and Shao, Changpeng and Chakraborty, Shantanav},
  title         = {Randomized {Quantum Singular Value Transformation}},
  year          = {2025},
  eprint        = {2510.06851},
  archivePrefix = {arXiv},
  primaryClass  = {quant-ph}
}

\end{document}